\documentclass[final,1p,times,twocolumn]{elsarticle}

\usepackage{mathrsfs}
\usepackage{amsmath}
\usepackage{graphics}
\usepackage{txfonts}
\usepackage{graphicx}
\usepackage{subfigure}
\usepackage{setspace}
\usepackage{amsfonts,amssymb,graphicx,mathrsfs,epsfig,subfigure}
\usepackage{multirow}
\usepackage{booktabs}
\usepackage{fancyhdr,times}
\usepackage{epstopdf}

\usepackage[colorlinks=true]{hyperref}

\hypersetup{urlcolor=blue, citecolor=red}

\numberwithin{equation}{section}
\usepackage{lineno}

\biboptions{sort&compress}

\newtheorem{theorem}{Theorem}[section]

\newtheorem{definition}{Definition}
\newtheorem{lemma}{Lemma}[section]

\newproof{proof}{Proof}
\newdefinition{remark}{Remark}

\allowdisplaybreaks[4]

\begin{document}
	\begin{frontmatter}

		\title{\Large \bf Backward bifurcation, basic reinfection number and robustness of a SEIRE epidemic model with reinfection  \tnoteref{t1}}

		\tnotetext[t1]{This work is supported by the National Natural Science Foundation of China (U21A20206) and Natural Science Foundations of Henan (192102310089, 202300410045).}
		\author[els]{Shaoli Wang \corref{cor1}}
		\ead{wslheda@163.com }
			\author[els]{Tengfei Wang }
		\author[els]{Ya-nen Qi }
		\author[wlu]{Fei Xu \corref{cor1}} \ead{fxu.feixu@gmail.com}

		\cortext[cor1]{Corresponding authors.}
		
		\address[els]{School of Mathematics and Statistics, Henan University, Kaifeng 475001, Henan, PR China }
		\address[wlu]{Department of Mathematics, Wilfrid Laurier University, Waterloo, Ontario, Canada \ N2L 3C5}

		\begin{abstract}
			
			\begin{spacing}{1.0}
				
				Recent evidences show that individuals who recovered from COVID-19 can be reinfected. However, this phenomenon has rarely been studied using mathematical models. In this paper, we propose a SEIRE epidemic model to describe the spread of the epidemic with reinfection.
				We obtain the important thresholds $R_0$ (the basic reproduction number) and  $R_c$ (a threshold less than one).
				Our investigations show that when $R_0>1$, the system has an endemic equilibrium, which is globally asymptotically stable.
When $R_c<R_0<1$, the epidemic system exhibits bistable dynamics. That is, the system has backward bifurcation  and the disease  cannot be eradicated.
				In order to eradicate the disease, we must ensure that the basic reproduction number $R_0$ is less than $R_c$. The basic reinfection number is obtained to measure the reinfection force, which turns out to be a new tipping
				point for disease dynamics. We also give definition of robustness,  a new concept to  measure the difficulty of completely eliminating the disease for a bistable epidemic system. Numerical simulations are carried out to verify the conclusions.
				
			\end{spacing}

		\end{abstract}
		

	\end{frontmatter}

	\section{Introduction}
	
Since ancient times, human beings have suffered from various epidemic diseases. As early as the 14th century, the plague epidemic caused 25 million deaths, and the European population was reduced by a quarter. The Spanish flu death toll from 1918 to 1920 exceeded 25 million. Since the 1981 AIDS pandemic, about 39 million people have died of the disease. Since the 1970s, new infectious diseases have been discovered almost every year. In the past 30 years, more than 40 new infectious diseases have appeared in the world, which has become a key and hot issue of global public health \cite{1}. At the beginning of 2020, the infection caused by the ``Novel Coronavirus" spread from Wuhan, a major city in central China, to the whole country coinciding with the peak of the Spring Festival travel season. In just one month, the number of infections exceeded 60,000, far exceeding the number of infections caused by Severe Acute Respiratory Syndrome Coronavirus (SARS-CoV) and Middle East Respiratory Syndrome Coronavirus (MERS-CoV) \cite{2}. Infectious diseases not only threaten human health and life, but also have a significant negative impact on the global economy. Therefore, research on infectious diseases is very meaningful.

	Bifurcation, especially the backward bifurcation, characterizes the dynamic behavior of many infectious disease models. \cite{20,8,Hadeler,Dushoff,Martcheva,14}.
Backward bifurcation reveals an important property in epidemic models, $R_0<1$ does not guarantee the eradication of  disease.
 In order to eradicate the disease, we must ensure that the basic reproduction number $R_0$ is less than $R_c$. Therefore, controlling such diseases is challenging, which has attracted a large number of scholars to explore in this area \cite{6,15,17,32,Zhang,7,Song}.

Although vaccination or recovery from infection provides immune protection, the interaction between  host and the carrier species will  reduce human immunity due to the complexity of carrier-borne diseases \cite{23}. Eventually a second infection will occur as the antibodies produced gradually diminish. Therefore, secondary infection in epidemiology has attracted the attention of scholars \cite{8,10,11,22}.

Heroin use is common in border areas. Some ex-addicts will start taking drugs again, creating a challenge to get rid of the addiction completely. In order to study drug use and formulate appropriate drug rehabilitation measures, scholars have established mathematical modeling to study drug addiction and drug rehabilitation problems. Many mathematical models have been constructed to address the relapse of drug addicts, indicating that relapse may occur in individuals with a history of drug use. Therefore, it is necessary to distinguish   the addiction rates of susceptible individuals with a history of drug use and those without a history of drug use. backward bifurcation is widely observed in such models \cite{24,25,26,27,28,29,30,31}.

Reinfection is also  discussed in tuberculosis (TB) exogenous models \cite{14,Song,Feng}, and sleeper effects models \cite{Song,Colon-Rentas}.	Recently, there have been reports about secondary infection with the novel coronavirus. There is evidence that antibody levels in COVID-19 patients gradually decline months after infection, making secondary infection possible \cite{33,34}. The number of people re-infected by the novel coronavirus is increasing, suggesting that for some people, immunity rapidly declines after contracting the virus.
 A MedRxiv study showed that severely ill patients infected with Covid-19 for the first time may develop ineffective antibodies and are more likely to develop serious secondary infections \cite{37,38}.
	
Based on the discussion above, we establish the following SEIRE model with re-infection under the assumption that all the infected individuals become exposed ones and the exposed is infectious.
	
	\begin{equation}\label{e1}
		\begin{cases} \frac{dS}{dt}=b-\mu S-\beta_{1} SE-\beta_{2} SI,\\
			\frac{dE}{dt}=\beta_{1} SE+\beta_{2} SI+\alpha_1 RE+\alpha_2 RI\\
			\hspace{0.7cm}-(\mu+k)E, \\
			\frac{dI}{dt}=k E-(\mu+\gamma)I,\\
			\frac{dR}{dt}=\gamma I-\mu R-\alpha_1 RE-\alpha_2 RI,
			\end {cases}
		\end{equation}
		with initial conditions
		\begin{equation}\label{e2}
			S(0)\geq0, E(0)\geq0, I(0)\geq0, R(0)\geq0.
		\end{equation}
		Here, $S,E,I,$ and $R$ are the numbers of susceptible, exposed, infected, and recovered individuals at time $t$. All the parameters in our model are positive. In the above model, $b$ is the constant recruitment rate in susceptible compartment only, $\mu$ is the natural mortality rate of each compartment, $\beta_{1}$ is the rate of susceptible individuals entering the exposed compartment due to contact with exposed individuals, $\beta_{2}$ is the rate of susceptible individuals entering the exposed compartment due to contact with infective individuals, $k$ is the rate of exposed individuals infected with disease into the infection compartment, $\gamma$ is the rate of infected individuals who have recovered through treatment, $\alpha_1$ is the rate of recovering individuals re-entering the exposed compartment due to contact with exposed individuals, and $\alpha_2$ is the rate of recovering individuals re-entering the exposed compartment due to contact with infected individuals.
		
		\begin{figure}[!h]
			\begin{center}
				{\rotatebox{0}{\includegraphics[width=0.4 \textwidth,
						height=30mm]{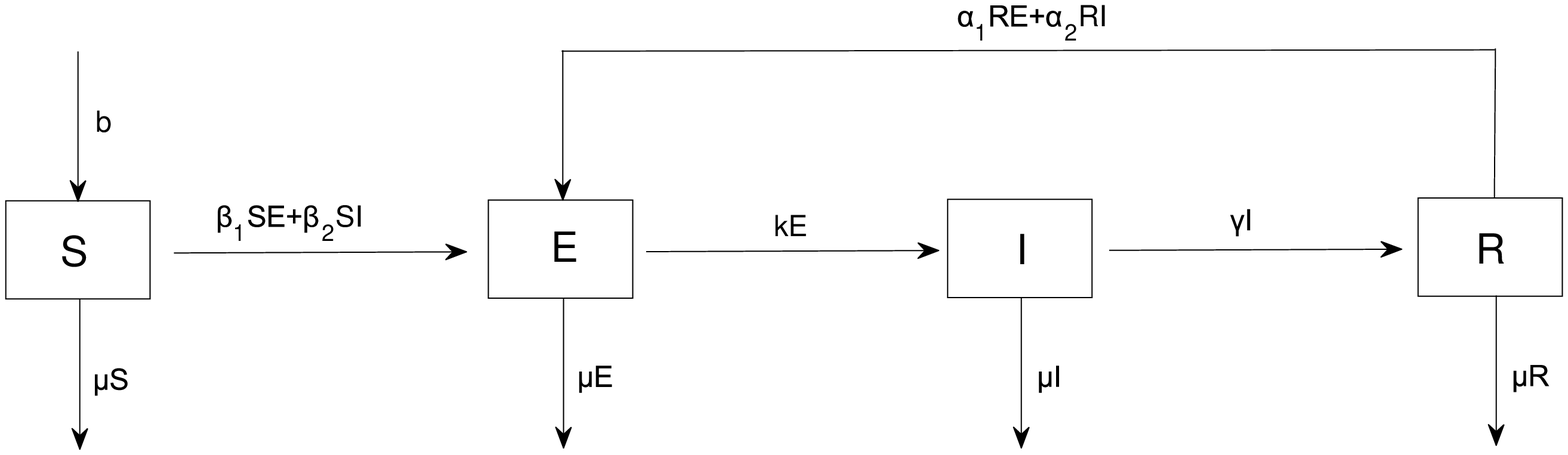}}}
				\caption{
					\footnotesize  The schematic diagram explaining the transmission dynamics of the epidemic disease.}\label{F51}
			\end{center}
		\end{figure}

		In order to facilitate the operation of the model, we do dimensionality reduction processing. The total population is denoted by $N=S+E+I+R. $ Then adding all the equations of  system \eqref{e1}, we obtain
		\begin{equation}\label{e3}
			\frac{dN}{dt}=b-\mu N.
		\end{equation}
		Suppose $N^{*}$ is the positive equilibrium of system \eqref{e3}, then $ N^{*}=\frac{b}{\mu}$. As the disease spreads, assume that the total population is in a stable demographic state, which gives $N^{*}=S+E+I+R$. So, $R$ can be replaced by $N^{*}-S-E-I$. Therefore, the system \eqref{e1} is simplified to the following three-dimensional system
		\begin{equation}\label{e4}
			\begin{cases}
				\frac{dS}{dt}=\mu N^{*}-\mu S-\beta_{1} SE-\beta_{2} SI, \\
				\frac{dE}{dt}=\beta_{1} SE+\beta_{2} SI\\
				\hspace{0.7cm}+\alpha_1 (N^{*}-S-E-I)E\\
				\hspace{0.7cm}	+\alpha_2 (N^{*}-S-E-I)I-(\mu+k)E, \\
				\frac{dI}{dt}=k E-(\mu+\gamma)I,\\
			\end{cases}
		\end{equation}
		with initial conditions \begin{equation}\label{e5}
			S(0)\geq0, E(0)\geq0, I(0)\geq0.
		\end{equation}

		The rest of this article is organized as follows:  In section \ref{sec3}, the basic properties of the solution of the model are discussed, including positivity and boundedness, and the basic reproduction number $R_0$ of the model is obtained by using the next-generation matrix method. In section \ref{sec4}, we discuss the existence of equilibria and analyze the stability of the system. In section \ref{sec5}, we discuss the existence of backward bifurcation in our model. In order to verify the analysis results obtained, numerical simulation is carried out in Section \ref{sec6}. We obtain the basic reinfection number and introduce the Robust of bistable system in Section \ref{sec55}. Last, we conclude the paper with discussions in Section \ref{sec7}.
		
		\section{Basic properties of the model}\label{sec3}
		
		\subsection{ Positivity and Boundedness of Solutions}
		
		\begin{theorem}\label{th3.1}
			Every solution of \eqref{e4} with positive initial conditions \eqref{e5} defined in $[0,\infty)$, will remain positive for all $t>0$.
		\end{theorem}
		
		\begin{proof}
			
			From the first equation of system \eqref{e4}, we get
			\begin{eqnarray*}
				\frac{dS(t)}{dt}=\mu N^{*}-\mu S(t)-\beta_{1} S(t)E(t)-\beta_{2} S(t)I(t)
				\\
				\geq-\mu S(t)-\beta_{1} S(t)E(t)-\beta_{2} S(t)I(t).
			\end{eqnarray*}
			Then,
			\begin{equation}\label{e6}
				\frac{dS(t)}{dt}+\left(\mu+\beta_{1} E(t)+\beta_{2} I(t)\right)S(t)\geq 0.
			\end{equation}
			Let $F(t)=\mu+\beta_{1} E(t)+\beta_{2} I(t)$. Multiplying both sides of inequality \eqref{e6}  by $exp\left(\int^t_0F(s)ds\right) $  yields
			$$
			exp\left(\int^t_0F(s)ds\right)\cdot\frac{dS(t)}{dt}+F(t)exp\left(\int^t_0F(s)ds\right)\cdot S(t)$$ $$\geq0,$$
			Then,
			\[
			\frac{d}{dt}\left(exp(\int^t_0F(s)ds)\cdot S(t)\right)\geq0.
			\]
			Integrating the above inequality from 0 to $t$, we get
			$$
			\int_0^t\frac{d}{ds}\left[exp\left(\int^s_0(\mu+\beta_{1} E(u)+\beta_{2} I(u))du\right)\cdot S(s)\right]ds$$
			$$\geq0.$$
			Then,
			\[
			S(t)\geq S(0)\cdot exp\left(-\int^t_0(\mu+\beta_{1} E(s)+\beta_{2} I(s))ds\right).
			\]
			Thus, we can get $S(t)>0$. Similarly, we can prove that $E(t)\geq0, I(t)\geq0$ and $R(t)>0.$
		\end{proof}
		\begin{theorem}\label{th3.2}
			Every solution of \eqref{e4} in $\mathbb{R}^4_+$ is bounded.
		\end{theorem}
		
		\begin{proof}
			From \eqref{e3}, we have
			\[
			\frac{dN(t)}{dt}=b-\mu N(t).
			\]
			Then we can get
			\[
			N(t)=\frac{b}{\mu}(1-e^{-\mu t})+N(0)e^{-\mu t}.
			\]
			Thus, we have
			\[
			\lim_{t\rightarrow\infty}N(t)\rightarrow\frac{b}{\mu}.
			\]
            Combined with the positivity of the solution, we can obtain the boundedness of the solution.
		\end{proof}
		
		\subsection{ The basic reproduction number}
		
		The threshold parameter $R_0$ gives the average number of infections transmitted by a single infected individual among fully susceptible individuals. To find $R_0$, we follow the next-generation matrix method proposed by van den Driessche and Watmough \cite{13}. Let us consider $X=(E, I, S)$ and rewrite system \eqref{e4} as $\frac{d\mathcal{X}}{dt}=\mathbb{F}-\mathbb{V}$, where $\mathbb{F}$ is the rate at which new infections occur, and $\mathbb{V}$ is all other traffic inside and outside of each compartments. So, we have
		$$\mathbb{F}=\left[
		\begin{array}{cccc}
			m\\
			0\\
			0\\
		\end{array}
		\right],$$
		where $m=\beta_1SE+\beta_2SI+\alpha_1 (N^*-S-E-I)E+ \alpha_2 (N^*-S-E-I)I,$
		and
		$$\mathbb{V}=\left[
		\begin{array}{cccc}
			(\mu+k)E\\
			(\mu+\gamma)I-kE\\
			-\mu N^*+\mu S+\beta_1SE+\beta_2SI\\
		\end{array}
		\right].$$
		
		The system \eqref{e4} always admits a disease-free equilibrium $Q_0=(N^*,0,0)$. Then, the jacobian matrices of $\mathbb{F}$ and $\mathbb{V}$ at $Q_0$ are given by
		$$\mathbb{DF_\mathrm{Q_0}}=\left[
		\begin{array}{cccc}
			F_{2\times2}     & 0\\
			0        & 0\\
		\end{array}
		\right]$$
		and
		$$\mathbb{DV_\mathrm{Q_0}}=\left[
		\begin{array}{cccc}
			V_{2\times2}     & 0\\
			M        & \mu\\
		\end{array}
		\right],$$
		where
		$$F=\left[
		\begin{array}{cccc}
			\beta_1 N^*     & \beta_2 N^*\\
			0        & 0\\
		\end{array}
		\right],$$
		$$V=\left[
		\begin{array}{cccc}
			\mu+k     & 0\\
			-k       &\mu+\gamma \\
		\end{array}
		\right] $$ and
		$$M=\left(\begin{array}{cccc}
			\beta_1 N^*    &\beta_2 N^*\\
		\end{array}
		\right).$$
		The form of the next generation matrix is
		\[
		FV^{-1}=\left[
		\begin{array}{cccc}
			\frac{\beta_1 N^*}{\mu+k}+\frac{\beta_2 kN^*}{(\mu+k)(\mu+\gamma)}     & \frac{\beta_2 N^*}{\mu+\gamma} \\
			0        & 0\\
		\end{array}
		\right].
		\]
		
		Now, according to Theorem 2 in  \cite{13}, the spectral radius $\rho$ of the matrix $FV^{-1}$ is the maximum eigenvalue of $FV^{-1}$, which gives the basic reproduction number $R_0$ of the system \eqref{e4}. Thus, we obtain
		\begin{equation*}\label{e7}
			R_0=\frac{\beta_1 N^*}{\mu+k}+\frac{\beta_2 kN^*}{(\mu+k)(\mu+\gamma)}.
		\end{equation*}

		\section{Existence and Stability analysis of equilibria}\label{sec4}
		
		\subsection{Existence of equibria}
		It is clear that system \eqref{e4} always admits a disease-free equilibrium $Q_0=(N^*,0,0)$. Then, we investigate the existence of the positive equilibrium $Q^*=(S^*,E^*,I^*)$ of system \eqref{e4}. In order to find the existence conditions of $Q^*$, we need to solve the following equations
		\begin{equation}\label{e9}
			\begin{cases}
				\mu N^*-(\mu+\beta_1 E^*+\beta_2 I^*)S^*=0, \\
				\beta_1 S^*E^*+\beta_2 S^*I^*\\
				\hspace{0.3cm}+\alpha_1 (N^*-S^*-E^*-I^*)E^*\\
				\hspace{0.3cm}+\alpha_2 (N^*-S^*-E^*-I^*)I^*\\
				\hspace{0.3cm}-(\mu+k)E^*=0, \\
				kE^*-(\mu+\gamma)I^*=0.
			\end{cases}
		\end{equation}
		Solving the third equation of \eqref{e9} to get $E^*=\frac{\mu+\gamma}{k}I^*$ and substituting the value of $E^*$ into the first equation of \eqref{e9}, we obtain $S^*=\frac{\mu kN^*}{\mu k+[\beta_1 (\mu+\gamma)+\beta_2 k]I^*}$.  Finally, substituting the values of $S^*$ and $E^*$ into the second  equation of \eqref{e9}, we get a quadratic equation about $I^*$ as follows
		\begin{equation}\label{e10}
			b_2 (I^*)^2+b_1 I^*+b_0=0,
		\end{equation}
		where
		\begin{eqnarray*}
			b_{2}&=&[\beta_1 (\mu+\gamma)+\beta_2 k]\{k^2\alpha_2\\
			&&+(\mu+\gamma)[k(\alpha_1+\alpha_2)+\alpha_1(\mu+\gamma)]\}, \\
			b_{1}&=&\mu k^3\alpha_2+\mu k\alpha_1 (\mu+\gamma)^2\\
			&&+\mu k^2(\mu+\gamma)(\alpha_1+\alpha_2) \\
			&&-k(\mu+\gamma)[\alpha_1N^*-(\mu+k)]\\
			&&\cdot[\beta_1(\mu+\gamma)+\beta_2 k]\\
			&&-k^2\alpha_2N^*[\beta_1(\mu+\gamma)+\beta_2 k],\\
			b_{0}&=&\mu k^2(\mu+k)(\mu+\gamma)(1-R_0).
		\end{eqnarray*}
		Obviously, the number of positive roots of polynomial \eqref{e10} depends on the signs of $b_0, b_1$ and $b_2$. This can be analyzed by applying Descarte's rule of sign. The various possibilities has been shown in Table \ref{tab1}.

		\begin{table*}[ht]
			\caption{Number of possible positive roots of polynomial $Eq.\eqref{e10}$ .  }\label{tab1}
			\begin{center}
				\begin{tabular}{l|llll|l}
					\hline    Cases & $b_2$ & $b_1$   & $b_0$ & ~$R_0$  & Total possible positive roots
					\\\hline
					~~~~$1$    &+ &-  &-    &$R_0>1$  & ~~~~~~~~~~~~~~~~~~~~~$1$\\
					~~~~$2$    &+ &+  &-    &$R_0>1$  & ~~~~~~~~~~~~~~~~~~~~~$1$\\
					~~~~$3$    &+ &-  &0    &$R_0=1$  & ~~~~~~~~~~~~~~~~~~~~~$1$\\
					~~~~$4$    &+ &+  &0    &$R_0=1$  & ~~~~~~~~~~~~~~~~~~~~~$0$\\
					~~~~$5$    &+ &+  &+    &$R_0<1$  & ~~~~~~~~~~~~~~~~~~~~~$0$\\
					~~~~$6$    &+ &-  &+    &$R_0<1$  & ~~~~~~~~~~~~~~~~~$0,1,2$\\
					\hline
				\end{tabular}
			\end{center}
		\end{table*}
		
		From the sixth case in Table \ref{tab1}, we know that the total number of positive roots of the polynomial \eqref{e10} depends on the sign of the discriminant $\Delta=b_1^2-4b_0b_2$ \cite{32}. From $\Delta=0$, we get
		$$R_0=1-\frac{b_1^2}{4b_2\mu k^2(\mu+k)(\mu+\gamma)}\triangleq R_c. $$
		Thus, we get the following lemma
		\begin{lemma}
			\begin{eqnarray*}
				\Delta>0\Leftrightarrow R_c<R_0,\\ \Delta=0\Leftrightarrow R_c=R_0,\\ \Delta<0\Leftrightarrow R_c>R_0.
			\end{eqnarray*}
		\end{lemma}

		To summarize, we have the following results on the existence of equilibria of \eqref{e4}.
		
		\begin{theorem}\label{th4.2}
			The system \eqref{e4}
			\begin{itemize}
				\item [{\rm (1)}] always admits a disease-free equilibrium $Q_0=(N^*,0,0)$,
				\item [{\rm (2)}] has a  unique endemic equilibrium $Q_+^*=(S_+^*,E_+^*,I_+^*)$ when $R_0>1$ and case 1 or 2 is satisfied,
				\item [{\rm (3)}]  has a  unique endemic equilibrium $Q_+^*=(S_+^*,E_+^*,I_+^*)$ when $R_0=1$ and case 3 is satisfied,
				\item [{\rm (4)}]  does not have any endemic equilibrium when $R_0=1$ and case 4 is satisfied,
				\item [{\rm (5)}] does not have any endemic equilibrium when $R_0<1$ and case 5 is satisfied,
				\item [{\rm (6)}] has one or more than one endemic equilibria when $R_0<1$ and case 6 is satisfied,
				\begin{itemize}
					\item  [{\rm (i)}] does not have any endemic equilibrium when $1>R_c>R_0$,
					\item  [{\rm (ii)}]	 has a unique endemic equilibrium $Q_*=(S_*,E_*,I_*)=(\frac{\mu kN^*}{\mu k+[\beta_1 (\mu+\gamma)+\beta_2 k]I_*},\frac{\mu+\gamma}{k}I_*,I_*)$ where $ I_*=\frac{-b_1}{2b_2} $ when $R_0=R_c$,
					\item  [{\rm (iii)}]has two endemic equilibria $Q_+^*=(S_+^*,E_+^*,I_+^*)$ and $Q_-^*=(S_-^*,E_-^*,I_-^*)$ when $R_C<R_0<1$.
				\end{itemize}
				\noindent Here,
				$Q_\pm^*=(S_\pm^*,E_\pm^*,I_\pm^*)=\left(\frac{\mu kN^*}{\mu k+[\beta_1 (\mu+\gamma)+\beta_2 k]I_\pm^*},\frac{\mu+\gamma}{k}I_\pm^*,\frac{-b_1\pm\sqrt{\Delta}}{2b_2}\right).$
			\end{itemize}
		\end{theorem}

		\subsection{Stability analysis of  the disease-free equilibrium }
		
		In order to study the local asymptotic stability of $Q_0$, we calculate  the Jacobian matrix of the system at $Q_0$. We then obtain
		$$J_{Q_0}=\left[
		\begin{array}{cccc}
			-\mu     &-\beta_1 N^*   &-\beta_2 N^*\\
			0        &\beta_1 N^*-(\mu+k)  &\beta_2 N^*\\
			0    &k   &-(\mu+\gamma)  \\
		\end{array}
		\right].$$
		Thus, the characteristic equation of the matrix $J_{Q_0}$ is given by
		\begin{equation}\label{e8}
			(\lambda+\mu)(\lambda^2+a_1 \lambda+a_0)=0,
		\end{equation}
		where $a_1=2\mu+\gamma+k-\beta_1 N^*$ and $a_0=(\mu+\gamma)(\mu+k-\beta_1 N^*)-\beta_2 kN^*=(\mu+k)(\mu+\gamma)(1-R_0).$
		
		All roots of  Eq.\eqref{e8} have negative real parts only when $a_0>0$ and $a_1>0$. It can be noted that $a_0>0$ if and only if $R_0<1$.  When $\beta_1<\frac{2\mu+\gamma+k}{N^*}$, $a_1>0$.
Therefore, all the eigenvalues of Jacobian $J_{Q_0}$ have negative real parts if $R_0<1$ and $\beta_1<\frac{2\mu+\gamma+k}{N^*}.$
		The results discussed above can be explained by the following theorem.
		
		\begin{theorem}\label{th4.1}
			The disease-free equilibrium $Q_0$ of the system \eqref{e4} is locally asymptotically stable only when $R_0<1$ and $\beta_1<\frac{2\mu+\gamma+k}{N^*};$ otherwise, it is unstable.
		\end{theorem}
		
		\subsection{Stability analysis of endemic equilibrium}
		
		\subsubsection{  Local asymptotic stability}
		
		To study the local asymptotic stability of the endemic equilibrium $Q^*$, we compute the following Jacobian matrix at $Q^*$, which is given by
		\begin{equation*}
			J_{Q^*}=\left[
			\begin{array}{cccc}
				J_{11}     &-\beta_1 S^*   &-\beta_2 S^*\\
				J_{21}      &J_{22}  &J_{23}\\
				0    &k   &-(\mu+\gamma)  \\
			\end{array}
			\right],
		\end{equation*}
		where $$J_{11}=-\mu-\beta_1 E^*-\beta_2 I^*,$$ $$J_{21}=(\beta_1-\alpha_1) E^*+(\beta_2-\alpha_2)I^*,$$ $$J_{22}=-[\beta_2S^*+\alpha_2(N^*-S^*-I^*)]\frac{I^*}{E^*}-\alpha_1E^*,$$
		$$J_{23}=\alpha_2N^*+(\beta_2-\alpha_2) S^*-(\alpha_1+\alpha_2)E^*-2\alpha_2 I^*.$$
		The characteristic equation of $J_{Q^*}$ is
		\[
		\lambda^3+c_1 \lambda^2+c_2 \lambda+c_3=0,
		\]
		where
		\begin{eqnarray*}
			c_1&=&2\mu+\gamma+(\beta_1+\alpha_1)E^*+\beta_2I^*\\
			&&+[\beta_2S^*+\alpha_2(N^*-S^*-I^*)]\frac{I^*}{E^*}>0, \\
			c_{2}&=&(\mu+\gamma)(\mu+\beta_1E^*+\beta_2I^*)\\
			&&+(2\mu+\gamma+\beta_1E^*+\beta_2I^*)\{\alpha_1E^*\\
			&&+[\beta_2S^*+\alpha_2(N^*-S^*-I^*)]\frac{I^*}{E^*}\}\\
			&&+k[(\alpha_1+\alpha_2)E^*+2\alpha_2I^*\\
			&&+(\alpha_2-\beta_2)S^*-\alpha_2N^*]\\
			&&+\beta_1S^*[(\beta_1-\alpha_1)E^*+(\beta_2-\alpha_2)I^*], \\
			c_3&=&\beta_1S^*(\mu+\gamma)[(\beta_1-\alpha_1)E^*+(\beta_2-\alpha_2)I^*]\\
			&&+k\beta_2S^*[(\beta_1-\alpha_1)E^*+(\beta_2-\alpha_2)I^*]\\
			&&+k(\mu+\beta_1E^*+\beta_2I^*)[(\alpha_1+\alpha_2)E^*\\
			&&+2\alpha_2I^*+(\alpha_2-\beta_2)S^*-\alpha_2N^*]\\
			&&+(\mu+\gamma)(\mu+\beta_1E^*+\beta_2I^*)\{\alpha_1E^*\\
			&&+[\beta_2S^*+\alpha_2(N^*-S^*-I^*)]\frac{I^*}{E^*}\}.
		\end{eqnarray*}
When all eigenvalues of $J_{Q^*}$  have negative real parts, the endemic equilibrium point $Q^*$ is locally asymptotically stable.
Therefore, using the well-known Routh-Hurwitz criteria, we obtain a set of parametric conditions for local asymptotic stability of $Q^*$,  given by $c_1>0, c_3>0$ and $c_1c_2-c_3>0.$
		The result can be summarized in the following theorem
		\begin{theorem}\label{th4.3}
			If $c_3>0$ and $c_1c_2-c_3>0,$
			the endemic equilibrium $Q^*$ of system \eqref{e4} is locally asymptotically stable.
		\end{theorem}

		\subsubsection{Global asymptotic stability}

		In this section, we study the global asymptotic stability of the endemic equilibrium point $Q^*$ of system \eqref{e4}. It can be seen from Theorem \ref{th4.2} that the system \eqref{e4} may have multiple endemic equilibria  independent of $R_0<1$ or $R_0>1$. In addition, according to the previous study, it is found that a backward bifurcation occurs when $R_0<1$, which  shows that the local equilibrium is not globally asymptotically stable in this case. However, when $R_0>1$ (i.e., case (i) of Theorem \ref{th4.2}), it is necessary  to study the overall stability of the local equilibrium point. In order to study the global asymptotic stability of $Q^*$, we will use the geometric method developed by Li and Muldowney \cite{18}. Now, we will briefly summarize the method developed by Li and Muldowney \cite{18}.

		Let us consider the mapping $x\rightarrow f(x)$ defined on an open set $\Omega\subset\mathbb{R}^n\rightarrow\mathbb{R}^n$ such that each solution of the differential equation
		\begin{equation}\label{eq11}
			\frac{dx}{dt}=f(x)
		\end{equation}
		is uniquely determined by its initial value $x(0)=x_0$, and the solution can be denoted by $x(t,x_0)$. Further, the following assumptions hold
		\begin{itemize}
			\item ($ H1 $) $\Omega$ is simply connected,
			\item  ($ H2 $) there is a compact absorbing set $E\subset\Omega$,
			\item ($ H3 $) the differential equation has an unique endemic equilibrium $x^*$.
		\end{itemize}
		The Lozinskii measure for an $n\times n$ matrix $B$ with respect to induced matrix norm $|\cdot|$ is defined as
		\[
		\eta(B)=\lim_{h\rightarrow 0^+}\frac{|I+hB|-1}{h}.
		\]
		Let us consider the map $x\rightarrow P(x)$, where $P(x)$ is a nonsingular matrix-valued $C_1$ function on $\Omega$.
The matrix $B$ is defined as $B=P_fP^{-1}+PV^{[2]}P^{-1}$, where $P_f$ is obtained by replacing each entry $p_{ij}$ of $P$ by its derivative in the direction of $ f $ and $V^{[2]}$ is the second additive compound matrix corresponding to the variational matrix $V$ of the system \eqref{eq11}. For the Lozinskii measure $\eta$ on $\mathbb{R}^{C_2\times C_2}$, a quantity is defined as
		$$q=\limsup_{t\rightarrow\infty}\sup_{x_0\in E}\frac{1}{t}\int^T_0\eta(Bx(s,x_0))ds. $$
		The following result has been established in Theorem 3.5 of \cite{18}.
		
		\begin{theorem}\label{th4.4}
			If  system \eqref{eq11} satisfies the assumptions ($ H1 $), ($ H2 $) and ($ H3 $), then the unique equilibrium $x^*$ is globally asymptotically stable in $\Omega$ when $q<0$  for a function $P(x)$ and Lozinskii measure $\eta$.
		\end{theorem}
		
		Now, we use Theorem \ref{th4.4} to investigate the global asymptotic stability of the infected equilibrium $Q^*$ for $R_0>1$. Before we start the proof, we claim that the system \eqref{e4} is uniformly persistent  by using the result demonstrated by Freedman et al.
		\begin{definition}
			The system \eqref{e4} is said to be uniformly persistent if there exists a constant $m>0$ such that any solution $(S(t),E(t),I(t))$ starting from $(S(0),E(0),I(0))\in \Gamma$ satisfies
			$$\min\{\liminf_{t\rightarrow\infty}S(t),\ \liminf_{t\rightarrow\infty}E(t),\ \liminf_{t\rightarrow\infty}I(t)\}\geq m.$$
		\end{definition}
		\begin{lemma}\label{le4.1}
			The system \eqref{e4} is uniformly persistent if and only if $R_0>1$.
		\end{lemma}
		
		The infection-free equilibrium point $Q_0$ is not locally asymptotic stable when $R_0>1$, which serves the necessity condition $R_0>1$. To prove that $R_0>1$ is sufficient for uniform persistent, we shall follow the approach described by Freedman in \cite{19}. To confirm that system \eqref{e4} satisfies all the conditions of Theorem 4.3 in \cite{19}, we consider $X=\mathbb{R}^3$ and $E=\Gamma$. The maximal invariant set $N$ on the boundary $\partial\Gamma$ is the disease-free equilibrium $Q_0$, which is isolated. Therefore, we may conclude from Theorem 4.3 in \cite{19} that the uniform persistence of \eqref{e4} when $R_0>1$ is equivalent to the instability of $Q_0$.
		
		Based on the above discussion, we establish the following theorem.
	
		\begin{theorem}
			The  unique endemic equilibrium $Q^*$ is globally asymptotically stable for $R_0>1$.
		\end{theorem}

		\begin{proof}
			
			System \eqref{e4} is uniformly persistent in the interior of simply connected domain $\Gamma$ when $R_0>1$. Therefore, there exits a compact absorbing set $E\subset int\Gamma$. Hence,  system \eqref{e4} satisfies the assumption ($ H2 $). Also, from the first case of Theorem \ref{th4.2},  we get the condition for the existence of a unique endemic equilibrium when $R_0>1$. Therefore, the assumption ($ H3 $) is also satisfied.

			The variational matrix $V(S,E,I)$ corresponding to the system \eqref{e4} is
			
			\begin{equation}\label{eq12}
				V=\left[
				\begin{array}{cccc}
					-\mu-\beta_1 E-\beta_2 I     &-\beta_1 S   &-\beta_2 S\\
					v_{21}     &v_{22}  &v_{23}\\
					0    &k   &-(\mu+\gamma)
				\end{array}
				\right],
			\end{equation}
			where
			\begin{eqnarray*}
				v_{21}&=&(\beta_1-\alpha_1)E+(\beta_2-\alpha_2) I,\\
				v_{22}&=&\alpha_1N^*+(\beta_1-\alpha_1)S-2\alpha_1 E\\
				&&-(\alpha_1+\alpha_2)I-(\mu+k),\\
				v_{23}&=&\alpha_2N^*+(\beta_2-\alpha_2)S\\
				&&-(\alpha_1+\alpha_2)E-2\alpha_2 I.
			\end{eqnarray*}
			The associated second additive compound matrix is
			\begin{equation}\label{eq13}
				V^{[2]}=\left[
				\begin{array}{cccc}
					v_1     &v_2   &\beta_2 S\\
					k        &-\beta_1 E-\beta_2 I-(2\mu+\gamma)  &-\beta_1 S\\
					0    &(\beta_1-\alpha_1)E+(\beta_2-\alpha_2)I   &v_3
				\end{array}
				\right],
			\end{equation}
			where
			\begin{eqnarray*}
				v_1&=&\alpha_1N^*+(\beta_1-\alpha_1)S-(2\alpha_1+\beta_1)E\\
				&&-(\alpha_1+\alpha_2+\beta_2)I-(2\mu+k),\\
				v_2&=&\alpha_2N^*+(\beta_2-\alpha_2)S\\
				&&-(\alpha_1+\alpha_2)E-2\alpha_1I,\\
				v_3&=&\alpha_1N^*+(\beta_1-\alpha_1)S-2\alpha_1E\\
				&&-(\alpha_1+\alpha_2)I-(2\mu+k+\gamma).
			\end{eqnarray*}
			Let us assume that the function $x\rightarrow P(x)$ as $P(S,E,I)=diag(1,\frac{E}{I},\frac{E}{I})$. Therefore, we have
			\begin{eqnarray*}
				&&P^{-1}(S,E,I)=diag(1,\frac{I}{E},\frac{I}{E}), \\
				&&P_f=diag(0,\frac{\dot{E}}{I}-\frac{E}{I^2}\dot{I},\frac{\dot{E}}{I}-\frac{E}{I^2}\dot{I}),\\
				&&P_fP^{-1}=diag(0,\frac{\dot{E}}{E}-\frac{\dot{I}}{I},\frac{\dot{E}}{E}-\frac{\dot{I}}{I}),
			\end{eqnarray*}
			and \begin{eqnarray*}
				B&=&P_fP^{-1}+PV^{[2]}P^{-1}\\
				&=&P_fP^{-1}+V^{[2]}\\
				&=&\left[
				\begin{array}{cccc}
					B_{11}     & B_{12}\\
					B_{21}       &B_{22}\\
				\end{array}
				\right],
			\end{eqnarray*}
			where
			\begin{eqnarray*}
				B_{11}&=&[\alpha_1N^*+(\beta_1-\alpha_1)S-(2\alpha_1+\beta_1)E\\
				&&-(\alpha_1+\alpha_2+\beta_2)I-(2\mu+k)], \\
				B_{12}&=&[\alpha_2N^*+(\beta_2-\alpha_2)S\\
				&&-(\alpha_1+\alpha_2)E-2\alpha_2I, \beta_2 S], \\
				B_{21}&=&[k,0]^T,\\
				B_{22}&=&\left[
				\begin{array}{cccc}
					c_1  & -\beta_1 S\\
					(\beta_1-\alpha_1)E+(\beta_2-\alpha_2) I       &c_2\\
				\end{array}
				\right].
			\end{eqnarray*}
			Here, $c_1=\frac{\dot{E}}{E}-\frac{\dot{I}}{I}-\beta_1 E-\beta_2 I-(2\mu+\gamma)$ and $ c_2=\frac{\dot{E}}{E}-\frac{\dot{I}}{I}+\alpha_1N^*+(\beta_1-\alpha_1)S-2\alpha_1E-(\alpha_1+\alpha_2)I-(2\mu+k+\gamma). $
			
			Now, we consider the norm on $\mathbb{R}^3$, obtained as  $$|(u,v,w)|=\max\{|u|, |v|+|w|\},  \forall  (u,v w)\in\mathbb{R}^3. $$
			And, the Lozinskii measure is defined as $$\eta(B)\leq\max\{g_1,g_2\}$$ with $$g_1=\eta_1(B_{11})+|B_{12}|$$ and $$g_2=\eta_1(B_{22})+|B_{21}|,$$ where $\eta_1$ is the Lozinskii measure of matrix with respect to the $L_1$ norm, and $|B_{12}|$ and $|B_{21}|$ are matrix norms with respect to $L_1$ vector norm. Therefore, we obtain
			\begin{equation}\label{eq14}
				\begin
				{array}{l l l l}
				|B_{12}|&=&\beta_2S+\max\{\alpha_2N^*-\alpha_2S\\
				&&-(\alpha_1+\alpha_2)E-2\alpha_2I, 0\}, \\
				|B_{21}|&=&\max\{k,0\}=k, \\
				\eta_1(B_{11})&=&\alpha_1N^*+(\beta_1-\alpha_1)S-(2\alpha_1+\beta_1)E\\
				&&-(\alpha_1+\alpha_2+\beta_2)I-(2\mu+k), \\
				\eta_1(B_{22})&=&\frac{\dot{E}}{E}-\frac{\dot{I}}{I}-\alpha_1 E-\alpha_2I-(2\mu+\gamma)\\
				&&+\max\{0, \alpha_1(N^*-S-E-I)-k\}.
				\end {array}
			\end{equation}
			Now, from the third equation of system \eqref{e4} , we obtain
			\begin{equation}\label{eq15}
				\frac{\dot{I}}{I}=k\frac{E}{I}-(\mu+\gamma).
			\end{equation}
			Therefore, from \eqref{eq14} and \eqref{eq15},  we obtain
			\begin{equation}\label{e16}
				\begin
				{array}{l l l l}
				\eta_1(B_{22})&=&\frac{\dot{E}}{E}-k\frac{E}{I}-\mu-\alpha_1E-\alpha_2I\\
				&&+\max\{0, \alpha_1(N^*-S-E-I)-k\}.
				\end {array}
			\end{equation}
			Hence, using the relations \eqref{e16} and \eqref{eq14} , we get
			\begin{equation*}
				\begin
				{array}{l l l l}
				g_2&=&\frac{\dot{E}}{E}-k\frac{E}{I}-\mu-\alpha_1E-\alpha_2I\\
				&&+\max\{0, \alpha_1(N^*-S-E-I)-k\}+k.
				\end {array}
			\end{equation*}
			Again, from the second equation of system \eqref{e4}, we get
			\begin{equation}\label{e17}
				\begin
				{array}{l l l l}
				\frac{\dot{E}}{E}&=&\beta_1 S+\beta_2 \frac{SI}{E}+\alpha_1(N^*-S-E-I)\\
				&&+\alpha_2(N^*-S-I)\frac{I}{E}-\alpha_2 I-(\mu+k).
				\end {array}
			\end{equation}
			Therefore, using this relations \eqref{e17} and \eqref{eq14} , we can rewrite $g_1$ as
			\begin{equation}
				\begin
				{array}{l l l l}
				g_1&=&\frac{\dot{E}}{E}-[(\beta_2-\alpha_2)S+\alpha_2(N^*-I)]\frac{I}{E}\\
				&&-(\alpha_1+\beta_1)E-\beta_2I+\beta_2S-\mu\\
				&&+\max\{\alpha_2N^*-\alpha_2S-(\alpha_1+\alpha_2)E-2\alpha_2I, 0\}.
				\end {array}
			\end{equation}

			Then, we can get
			$$\eta(B)\leq\max\{g_1,g_2\}=\frac{\dot{E}}{E}-(\mu-\theta),$$
			where $\theta=\max\{\theta_1+k-k\frac{E}{I}-\alpha_1E-\alpha_2I, \theta_2+\beta_2(S-I)-[(\beta_2-\alpha_2)S+\alpha_2(N^*-I)]\frac{I}{E}-(\alpha_1+\beta_1)E\}$, where $\theta_1=\max\{0, \alpha_1(N^*-S-E-I)-k\},$ and $\theta_2=\max\{\alpha_2N^*-\alpha_2S-(\alpha_1+\alpha_2)E-2\alpha_2I, 0\}$. Finally, we obtain
			\begin{eqnarray*}
				q&=&\frac{1}{t}\int^t_0\eta(Bx(s,x_0))ds\\
				&\leq&\frac{1}{t}\int^t_0\frac{\dot{E}}{E}ds-(\mu-\theta)\\
				&=&\frac{1}{t}\ln\frac{E(t)}{E(0)}-(\mu-\theta),
			\end{eqnarray*}
which implies that
			$$\Rightarrow\lim_{t\rightarrow\infty}\sup_{x_0\in E}\frac{1}{t}\int^t_0\eta(Bx(s,x_0))ds\leq0,\ \text{ if } \ \mu>\theta.$$
			Therefore, we can conclude that the infected equilibrium, when it exits uniquely, is globally asymptotically stable for $R_0>1$.
		\end{proof}
		\section{ Backward bifurcation}\label{sec5}
			In epidemiological models, the occurrence of backward bifurcation is an important phenomenon. Backward bifurcation in disease models have been studied by many scholars \cite{6,7,14,15,17}. In our model \eqref{e4}, there are multiple disease persistent equilibria $Q^*$ for $R_0<1$, which indicates the possibility of backward bifurcation. %
Epidemiologically, the value of $R_0$ is not sufficient to determine whether the disease will persist. When $R_0<1$, the future state of the epidemic depends on the initial size of individuals.
Our purpose is to study the existence value of the backward bifurcation in \eqref{e4}. Here, we use the famous results of Castillo--Chavez and Song \cite{14}.
		
		We simplify  system \eqref{e4} and choose $S=x_1,E=x_2,I=x_3$. If we set $X=(x_1,x_2,x_3)^T,$ then our system \eqref{e4} can be written in the form $\frac{dX}{dt}=F(X)$ with $F=(f_1,f_2,f_3)^T$, where
		\begin{eqnarray*}
			\left[
			\begin{array}{cc}
				f_1\\
				f_2\\
				f_3\\
			\end{array}
			\right]=\left[
			\begin{array}{cccc}
				f_{11} \\
				f_{22} \\
				f_{33}
			\end{array}
			\right],
		\end{eqnarray*}
		where
		\begin{equation}
			\begin
			{array}{l l l l}
			f_{11}&=&\mu N^*-\mu x_1-\beta_1 x_1x_2-\beta_2 x_1x_3,\\
			f_{22}&=&\beta_1 x_1x_2+\beta_2 x_1x_3\\
			&&+\alpha_1(N^*-x_1-x_2-x_3)x_2\\
			&&+\alpha_2(N^*-x_1-x_2-x_3)x_3\\
			&&-(\mu+k)x_2,\\
			f_{33}&=&kx_2-(\mu+\gamma)x_3.
			\end {array}
		\end{equation}
		Then, we can get the Jacobian matrix of the system at the disease-free equilibrium point $Q_0=(N^*,0,0)$ as follows
		\begin{eqnarray*}
			J_{Q_0}=\left[
			\begin{array}{cccc}
				-\mu     &-\beta_1 N^*   &-\beta_2 N^*\\
				0        &\beta_1 N^*-(\mu+k)  &\beta_2 N^*\\
				0    &k   &-(\mu+\gamma)
			\end{array}
			\right].
		\end{eqnarray*}
		
		Choosing $\beta_2$ as a bifurcation parameter, when $R_0=1$, we can obtain the critical value for $\beta_2=\beta_c=\frac{(\mu+\gamma)(\mu+k-\beta_1N^*)}{kN^*}$. In this case, the jacobian matrix $J_{Q_0}$ has a simple zero eigenvalue whose left and right eigenvectors are given by $v=(0,1,\frac{(\beta_1+\beta_2)N^*-(\mu+k)}{\mu+\gamma-k})$ and $w=(\frac{-(\mu+k)(\mu+\gamma)}{\mu k},\frac{\mu+\gamma}{k},1)^T$, respectively.

		To obtain the following quantities reported in Theorem 4.1 in \cite{14}, we have
		\begin{eqnarray*}
			\begin
			{array}{l l l l}
			a&=&\sum^3_{k,i,j=1}v_kw_iw_j\frac{\partial^2f_k}{\partial x_i \partial x_j}(Q_0,\beta_c),\\
			b&=&\sum^3_{k,i=1}v_kw_i\frac{\partial^2f_k}{\partial x_i \partial \beta_2}(Q_0,\beta_c).
			\end {array}
		\end{eqnarray*}
		It can be noted that the first component of $v$ is zero, so we do not need to find the partial derivative of $f_1$. Because the expression of $f_3$ is one-time, the second-order partial derivatives of $f_3$ are all zero. The non-zero partial derivative of $f_2$ can be written as
		\begin{eqnarray*}
			\begin
			{array}{l l l l}
			\frac{\partial^2f_2}{\partial x_1\partial x_2}&=&\beta_1-\alpha_1,~~ \frac{\partial^2f_2}{\partial x_1\partial x_3}=\beta_2-\alpha_2,\\
			\frac{\partial^2f_2}{\partial x_2\partial x_3}&=&-\alpha_1-\alpha_2,~~ \frac{\partial^2f_2}{\partial x_2^2}=-2\alpha_1,\\
			\frac{\partial^2f_2}{\partial x_3^2}&=&-2\alpha_2,~~ \frac{\partial^2f_2}{\partial x_1\partial \beta_2}=x_3=0, \\
			\frac{\partial^2f_2}{\partial x_3\partial \beta_2}&=&x_1=N^*.
			\end {array}
		\end{eqnarray*}
		
		Calculating the values of $a$ and $b$ at $(Q_0, \beta_c)$ yields
		\begin{eqnarray*}
			\begin
			{array}{l l l l}
			a&=&\frac{2(\mu+k)(\mu+\gamma)}{\mu k}[\frac{(\alpha_1-\beta_1)(\mu+\gamma)}{k}+\alpha_2-\beta_c]\\
			&&-\frac{2(\mu+\gamma)}{k}[\alpha_1+\alpha_2+\frac{\alpha_1(\mu+\gamma)}{k}]-2\alpha_2,\\
			b&=&v_2w_3N^*=N^*.
			\end {array}
		\end{eqnarray*}

		Thus, our system undergoes backward bifurcation at $\beta=\beta_c$, only when both $a$ and $b$ are positive at $(Q_0,\beta_c)$. Obviously, $b$ is always positive. Therefore, the positivity of $a$ gives the threshold condition for the backward bifurcation $$\alpha_2>\alpha^*=\frac{(\mu+\gamma)\{(\mu+k)^2(\mu+\gamma)-k\gamma\alpha_1 N^*\}}{\gamma k^2N^*}.$$

		The result can be summarized in the following theorem
		
		\begin{theorem}\label{thback}
			If  $\alpha_2>\alpha^*=\frac{(\mu+\gamma)\{(\mu+k)^2(\mu+\gamma)-k\gamma\alpha_1 N^*\}}{\gamma k^2N^*}$,  system \eqref{e4} will experience a backward bifurcation.
		\end{theorem}

			\section{Numerical simulations}\label{sec6}
			
			In this section, some numerical simulations are carried out to visualize the obtained analysis results.			
			
			In order to verify the discussion about backward bifurcation, we select a set of parameter values $N^*=60, k=0.02, \mu=0.013, \beta_1=0.0003$, $ \beta_2=0.0001$, $ \alpha_1=0.03$, $ \alpha_2=0.04$ and $\gamma=0.1$.  This set of parameters ensures that  $\alpha_2=0.04>\alpha^*=-0.1637$. Thus, $a$ and $b$ are both non-negative, which guarantees that  system \eqref{e4} experiences  a backward bifurcation. By numerical simulation, we get the bifurcation diagram  of system \eqref{e4} (see Figure \ref{F6}). It is clear   that  system \eqref{e4} has two endemic equilibria when $R_C<R_0<1$. The solid blue line above indicates a stable endemic equilibrium, while the red dotted line indicates an unstable endemic equilibrium. Also, whenever $R_0<1$, the disease-free equilibrium is locally asymptotically stable, represented by a solid blue line.

			We set the parameters $N^*=60, k=0.02, \mu=0.013$, $ \beta_1=0.0006$, $ \beta_2=0.0006$, $ \alpha_1=0.03$, $ \alpha_2=0.04$ and $\gamma=0.1$, which ensures that $R_0=1.2840>1$.  We choose different initial values to get the solution trajectory diagram of system \eqref{e4} (see Figure \ref{F2}). In this case, $b_2=4.4470\times10^{-8}>0, b_1=-2.5039\times10^{-7}<0,$ and $ b_0=-5.5068\times10^{-9}<0$. Therefore, according to case 1 of Table \ref{tab1},  system \eqref{e4} has an endemic equilibrium $Q_+^*=(21.9388, 31.9366, 5.6525)$. In this case, $c_1=1.3384>0$, $ c_2=0.1885>0$, $ c_3=0.0036>0$, and $ c_1c_2-c_3=0.2488>0$.  From Theorem \ref{th4.3}, we know that the local equilibrium $Q_+^*$ is locally asymptotically stable. Thus, $\frac{2\mu+\gamma+k}{N^*}=0.0024>\beta_1=0.0006$.  Since $R_0>1$,  according to Theorem \ref{th4.1}, the disease-free equilibrium $Q_0=(60, 0, 0)$ is unstable.
			
			In addition, when selecting parameters $\mu=0.011, \beta_1=0.0001, \beta_2=0.0003, \alpha_1=0.001, \alpha_2=0.001$ and $\gamma=0.001,$ we get $R_0=1.1613>1$. We choose different initial values to get the solution trajectory diagram of system \eqref{e4} (see Figure \ref{F3}). Then, we  get $b_2=7.3728\times10^{-12}>0, b_1=2.3680\times10^{-12}>0,$ and $ b_0=-2.6400\times10^{-10}<0$. Therefore, according to case 2 of Table \ref{tab1},  system \eqref{e4} also has an endemic equilibrium $Q_+^*=(50.3925, 3.4953, 5.8255)$. In this case, $c_1=0.0601>0$, $ c_2=8.7736\times10^{-4}>0$, $ c_3=1.2855\times10^{-6}>0$, $ c_1c_2-c_3=5.1437\times10^{-5}>0$. From Theorem \ref{th4.3}, the local equilibrium $Q_+^*$ is locally asymptotically stable. Under the above parameters,  $\frac{2\mu+\gamma+k}{N^*}=7.1667\times10^{-4}>\beta_1=0.0001$, so we conclude that the disease-free equilibrium $Q_0$ is unstable. This shows that the system \eqref{e4} has an endemic equilibrium and a disease-free equilibrium, where the endemic equilibrium is globally asymptotically stable and the disease-free equilibrium is unstable.
			
			Selecting $N^*=60, k=0.02, \mu=0.013$, $ \beta_1=0.0003$, $\beta_2=0.0001$, $ \alpha_1=0.03$, $ \alpha_2=0.03$ and $\gamma=0.1$,  we  get $1>R_0=0.5775>R_c=0.3396$. We also get $b_2=1.9051\times10^{-8}>0$, $ b_1=-3.1238\times10^{-8}<0$, $ b_0=8.1900\times10^{-9}>0$, and $ \bigtriangleup=b_1^2-4b_0b_2=3.5167\times10^{-16}>0$. According to case 6 of Table \ref{tab1},  system \eqref{e4} has two endemic equilibria which are $Q_+^*=(50.7976, 7.4129, 1.3120)$ and $Q_-^*=(57.4029, 1.8513, 0.3277)$. However, for $Q_-^*$, we have $c_3=-1.5362\times10^{-5}<0$, so the endemic equilibrium $Q_-^*$ is unstable. As for $Q_+^*$, we get $c_1=0.3935>0, c_2=0.0367>0, c_3=6.1510\times10^{-5}>0,$ and $ c_1c_2-c_3=0.0144>0$, so endemic equilibrium $Q_+^*$ is stable. From Theorem \ref{th4.1}, we know that the disease-free equilibrium $Q_0$ is locally asymptotically stable because $\frac{2\mu+\gamma+k}{N^*}=0.0024>\beta_1=0.0003$. In summary,  system \eqref{e4} has two endemic equilibria $Q_+^*, Q_-^*$ and a disease-free equilibrium $Q_0$ when $R_c<R_0<1$, where both $Q_+^*$ and $Q_0$ are locally asymptotically stable (see Figure \ref{F4}).

			Considering the case of $R_0<R_c<1$, we set the parameters $\mu=0.013$, $ \beta_1=0.0003$, $ \beta_2=0.0001$,  $ \alpha_1=0.01$,  $ \alpha_2=0.01$ and $\gamma=0.1$. Then, we get $R_0=0.5776<R_c=0.8489<1$. In this case, $b_2=6.3504\times10^{-9}>0$, $ b_1=-8.6276\times10^{-9}<0$, $ b_0=8.1900\times10^{-9}>0$, and  $  \bigtriangleup=b_1^2-4b_0b_2=-1.3360\times10^{-16}<0$, which satisfied the Case 6 of Table \ref{tab1}. Therefore, the system \eqref{e4} has only one disease-free equilibrium $Q_0=(60,0,0)$, and no endemic equilibrium. Also, $\frac{2\mu+\gamma+k}{N^*}=0.0024>\beta_1=0.0003$, so according to Theorem \ref{th4.1}, the disease-free equilibrium $Q_0$ is locally asymptotically stable. In this case, the solutions of \eqref{e4} with different initial values converge to $Q_0$ as shown in Figure \ref{F5}.
			

From the numerical simulations above, we find that the backward bifurcation and the existence of multiple equilibria complicate the dynamics of the model.
As shown in Figure \ref{F6}, we find that when $R_0$ crosses $1$, the number of infectious cases will suddenly rebound. In addition, when the system is in an epidemic state, slowly reducing $R_0$ to the critical value of $1$, we find that even if $R_0$ is slightly less than $ 1 $ and greater than $R_c$, the system may not return to the disease-free state, but still in the epidemic state. Therefore, $R_0<1$ does not ensure the eradication of the disease. Figure \ref{F6} shows that only when $R_0$ is less than $R_c$, the endemic equilibria disappear and the system converges to the infection-free steady state $Q_0$. Therefore, it can be concluded that $R_0<R_c<1$ is a sufficient condition for eradicating disease.
			
In the following, we use numerical simulation  to evaluate the effect of contact rates $\beta_1$ and $\beta_2$ on the threshold $ R_c $.
Here, we use the same parameters as in Figure \ref{F6} except for $\beta_1$ and $\beta_2$.
%
%
We appropriately reduce or increase $\beta_1$ and $\beta_2$. From Figure \ref{F7} (a), we can see that $R_c$ gradually decreases with the increase of $\beta_1$. When $\beta_1$ increases to 0.003129, $R_c=0$, which means the disease cannot be eradicated. We appropriately adjust the value of $\beta_2$ to obtain the situation shown in Figure \ref{F7} (b).  As $\beta_2$ increases, $ R_c $ gradually decreases. When $\beta_2$ increases to 0.000173, $R_c=0$, which means that the disease will persist and cannot be eradicated.
In addition, we notice that $\alpha_1$ and $\alpha_2$ in the model also have significant  influence on $ R_c $.

Figure \ref{F7}(C) shows the effect of changes in $\alpha_1$ on $R_c$. The rest of the parameters are the same as those of Figure \ref{F6}.
We find that $R_c$ gradually decreases with the increase of $\alpha_1$. When $\alpha_1=0.04368$, $R_c=0$, which means that the disease cannot be eradicated.
Similarly, we can use numerical simulations to study the effect of changes in $\alpha_2$ on $R_c$ (see Figure \ref{F7} (d)). We find that $ R_c $ will decreases as $\alpha_2$ increases and $ R_c =0$ when $ \alpha_2=0.1172 $.

			\section{Basic reinfection number and robustness of bistability }\label{sec55}
			
			The basic reinfection number \cite{Song} and the basic reproduction number \cite{13}  characterize the spread of infectious disease. Below, we will combine the basic reinfection number and the basic reproduction
number to give a complete disease control measure when there is a reinfection (or relapse).
%
The basic reinfection number is given by
			$$  R_r=\frac{\alpha_2\gamma k^2N^*}{(\mu+k)^2(\mu+\gamma)^2}+\frac{\alpha_1 \gamma k N^*}{(\mu+k)^2(\mu+\gamma)}, $$
			which is calculated from $ \alpha_2=\alpha^* $. Then, Theorem \ref{thback} can be rewritten as:
			\begin{theorem}
			If the basic reinfection number $ R_r>1 $, then   system \eqref{e4} will experience a backward bifurcation.
	    	\end{theorem}
			
%
As we all know, when studying the primary infection, we use the basic reproduction number $R_0 $ to measure the force of the primary infection.
Thus, corresponding to $ R_0 $, the basic reinfection number $ R_r $  measures the reinfection forces (or capability of relapse). If the basic reproduction number $ R_0 $ is greater than one, the primary infection will invade a population.
			In the range of $ R_0<1 $, if reinfection force is strong enough to make the basic reproduction number $ R_r> 1 $, the disease may be persistent. However, when the basic reproduction number $ R_r $ is too small, there are not enough recovered individuals to be reinfected, then the disease will disappear completely. Besides, the basic reproduction number $ R_r $ also  characterize the type of bifurcation when the basic reproduction number is equal to one. If the basic reinfection number is greater than one, the bifurcation is backward. Otherwise it is forward.

			Then, we define the robustness of bistable system \eqref{e4}, which is represented  by the definite integral of positive steady solution curve on interval $[R_c, 1]$
				\begin{equation*}
				\begin
				{array}{l l l l}
		 R&=&\int_{R_c}^{1}\left(\frac{-b_1+\sqrt{\Delta}}{2b_2}- \frac{-b_1-\sqrt{\Delta}}{2b_2}\right)dR_0\\
		 &=&\int_{R_c}^{1}\frac{\sqrt{\Delta}}{b_2}dR_0,		
		 	\end{array}
	 		\end{equation*}
			where $ \Delta=b_1^2-4b_2b_0 $ is the discriminant of polynomial \eqref{e10}. Based on the definition of robustness $R$ and the discussion about backward bifurcation, we give following theorem.
			\begin{theorem}
				If the robustness $ R>0 $, then system \eqref{e4} will experience a backward bifurcation.
			\end{theorem}
		
			The robustness of bistable system can be used to describe the system affected by the change of initial value. The values of  $ R $ with different values of $ \beta_1,~\beta_2,~\alpha_1  $ and $\alpha_2 $ are listed in Tables \ref{tab2} - \ref{tab5}. From the tables, we can see that the value of $ R $ increases with the increase of contact rate, that is, the higher the contact rate is,  the  stronger the robustness of bistable system is. The robustness $R$ can be used to express the difficulty of completely eliminating the disease. The larger the $R$, the stronger the robustness of the bistable system and the more difficult it is to eliminate the disease.

					\begin{table}[!ht]
						\centering
							\caption{The robustness of bistable system $ R $ with different values of $\beta_1$ (the rate of susceptible individuals entering the exposed compartment due to contact with exposed individuals). The other parameters are the same as those used in Figure \ref{F6}.}
				 \begin{tabular}{l|lll|l}
						\hline
						\hspace{0.2cm}$\beta_1$~~~    &  \hspace{1cm} $R_r $   &  \hspace{1.4cm} $ R $\\
						\hline
						0.00026 ~~~  &~~~~~~  36.1585 &  \hspace{1.1cm}   0.033  \\
						0.00028   ~~~ &~~~~~~  36.1585 & \hspace{1.1cm}  0.2514\\
						0.0003      ~~~ &~~~~~~  36.1585 & \hspace{1.1cm}  0.7674\\
						0.00031     ~~~   &~~~~~~  36.1585 &  \hspace{1.1cm}   1.155\\
						0.0003129    ~~~ &~~~~~~  36.1585 & \hspace{1.1cm} 1.2844 \\
						\hline
					\end{tabular}\label{tab2}
	            	\end{table}
            	
						\begin{table}[!ht]
						\centering
						\caption{The robustness of bistable system $ R $ with different values of $\beta_2$ (the
							rate of susceptible individuals entering the exposed compartment due to contact with infective individuals). The other parameters are the same as those used in Figure  \ref{F6}. }
						 \begin{tabular}{l|lll|l}
							\hline
						\hspace{0.2cm}	$\beta_2$~~~    &  \hspace{1cm} $R_r $   &  \hspace{1.4cm} $ R $\\
							\hline
							0.00001 ~~~  &~~~~~~  36.1585 & \hspace{1.1cm}  0.3299  \\
							0.00005  ~~~ & ~~~~~~  36.1585 & \hspace{1.1cm}  0.498\\
							0.0001      ~~~ &~~~~~~  36.1585 & \hspace{1.1cm}   0.7674\\
							0.00015     ~~~   & ~~~~~~ 36.1585  & \hspace{1.1cm}    1.1058\\
							0.000173    ~~~ & ~~~~~~ 36.1585  &\hspace{1.1cm} 1.2853 \\
							\hline
						\end{tabular}\label{tab3}
					\end{table}
					
						\begin{table}[!ht]
						\centering
						\caption{The robustness of bistable system $ R $ with different values of $\alpha_1$ (the rate of recovering individuals
							re-entering the exposed compartment due to contact
							with exposed individuals). The other parameters are the same as those used in Figure  \ref{F6}. }
					 \begin{tabular}{l|lll|l}
							\hline
						\hspace{0.2cm}	$\alpha_1$~~~  & \hspace{1cm} $R_r $   &  \hspace{1.4cm} $ R $\\
							\hline
							0.01 ~~~  & ~~~~~~  16.6554  &  \hspace{1.1cm}   0.2618  \\
							0.02   ~~~ &  ~~~~~~  26.4069  &  \hspace{1.1cm}  0.5115\\
							0.03      ~~~ &  ~~~~~~  36.1585 &  \hspace{1.1cm}   0.7674\\
							0.04     ~~~   & ~~~~~~   45.9101   & \hspace{1.1cm}   1.0257\\
							0.04368    ~~~ & ~~~~~~   49.4987  & \hspace{1.1cm} 1.1211 \\
							\hline
						\end{tabular}\label{tab4}
					\end{table}
					\begin{table}[!ht]
					\centering
					\caption{The robustness of bistable system $ R $ with different values of $\alpha_2$ (the rate of recovering individuals re-entering the exposed compartment due to contact with infected individuals). The other parameter values  are the same as those used in Figure  \ref{F6}. }
					\begin{tabular}{l|lll|l}
						\hline
					\hspace{0.2cm}	$\alpha_2$~~~   & \hspace{1cm} $R_r $   &  \hspace{1.4cm} $ R $\\
						\hline
						0.00001 ~~~  &~~~~~~   29.2565  &  \hspace{1.1cm}   0.5859  \\
						0.01   ~~~ & ~~~~~~   30.9807  & \hspace{1.1cm}  0.6311\\
						0.05      ~~~ & ~~~~~~  37.8845 &  \hspace{1.1cm}   0.813\\
						0.1     ~~~   &   ~~~~~~ 46.5142  &  \hspace{1.1cm}    1.0417\\
						0.1172    ~~~ &   ~~~~~~ 49.4828    &  \hspace{1.1cm} 1.1206 \\
						\hline
					\end{tabular}\label{tab5}
				\end{table}
			\section{Discussion}\label{sec7}
			
			In this article, we studied the SEIRE model of an infectious disease and developed a compartment model to study the transmission of the infection. In order to facilitate the calculation, we  reduced the dimension of the initial system to get system \eqref{e4}. We prove the positivity and boundedness of solutions for system \eqref{e4}. We get the basic reproduction number $R_0$. Then, we present  the existence conditions of equilibria and  their stability.
			
			Also, we find that under some conditions, the system will undergo backward bifurcation. This means that $R_0<1$ does not guarantee the eradication of the disease. 
Only when the system has no endemic equilibria, i.e., $ R_0<R_c$, the disease will be totally eradicated. 
The results suggest that disease rebound may occur even when the basic reproductive number is less than $1$.
In epidemic control, it is necessary to ensure that the basic infection number is far below 1 to completely control the epidemic.
Our analysis results were verified by numerical simulation. We found that the system exhibits bistability under certain conditions, and a backward bifurcation occurs.
%
We simulated the effect of $\beta_1, ~\beta_2, ~\alpha_1$ and $\alpha_2$ on the threshold $R_c$.
We find that reducing these parameters can increase $R_c$, implying that reducing the contact rates with exposed and infected individuals is beneficial to epidemic control. 
%
Lastly, we give the basic reinfection number $ R_c $ and the robustness $R$.

			\newpage

			\begin{figure}[!h]
				\begin{center}
					
					{\rotatebox{0}{\includegraphics[width=0.5\textwidth,
							height=40mm]{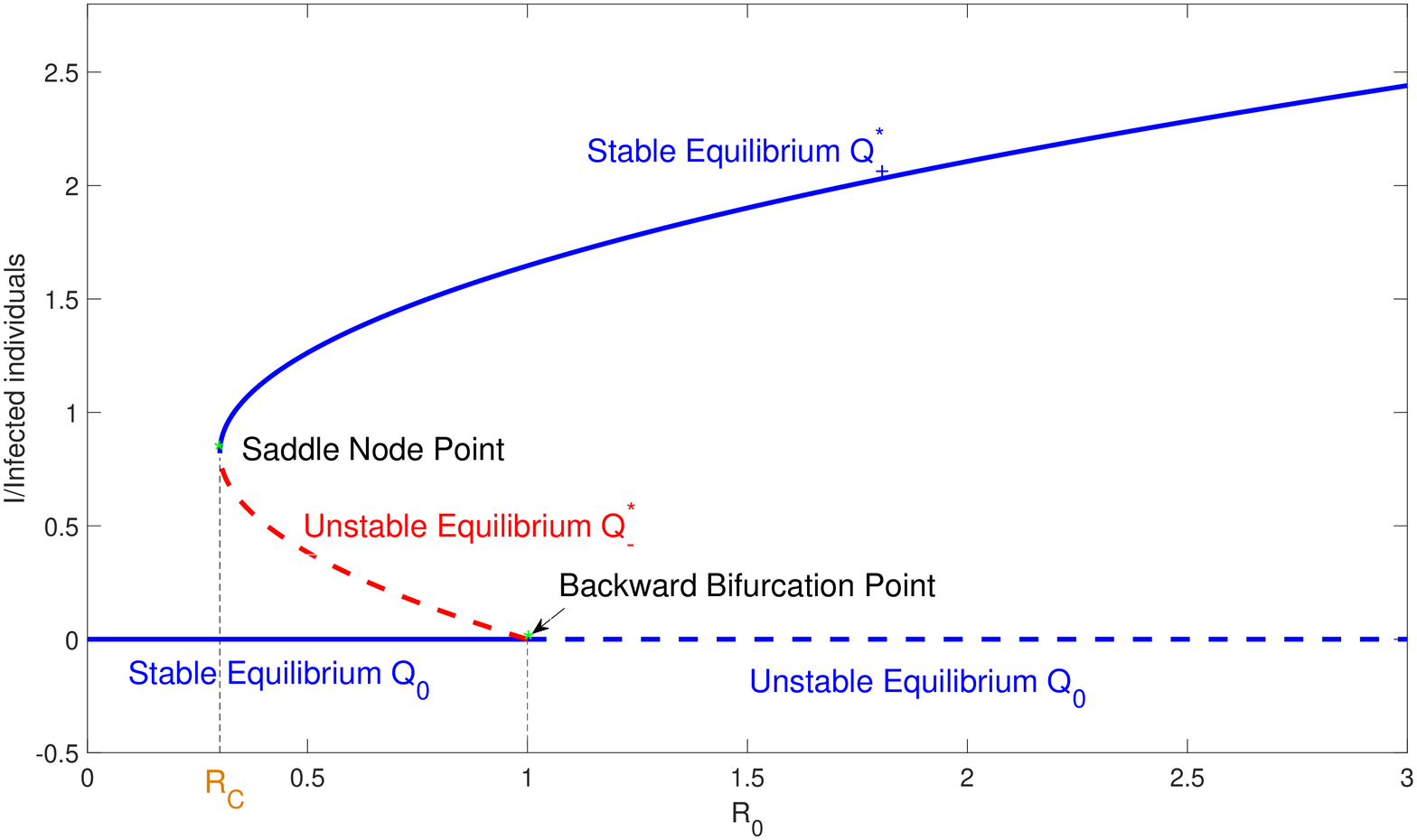}}}
					\caption{
						\footnotesize Bifurcation diagrams of system \eqref{e4}. Here, $N^*=60, \,k=0.02,\, \mu=0.013, \,~\beta_1=0.0003,\, ~\beta_2=0.0001,\, ~\alpha_1=0.03, \, ~\alpha_2=0.04,$ and $~\gamma=0.1$.  The upper solid blue line indicates the stable endemic equilibrium $ Q^*_+ $, and the lower solid blue line indicates the stable disease-free equilibrium $ Q_0 $; the red and blue dotted lines indicate the unstable endemic equilibrium $ Q^*_+ $ and unstable disease-free equilibrium $ Q_0 $, respectively.}\label{F6}
				\end{center}
			\end{figure}

			\begin{figure}[!h]
				\begin{center}
					{\rotatebox{0}{\includegraphics[width=0.48 \textwidth,
							height=40mm]{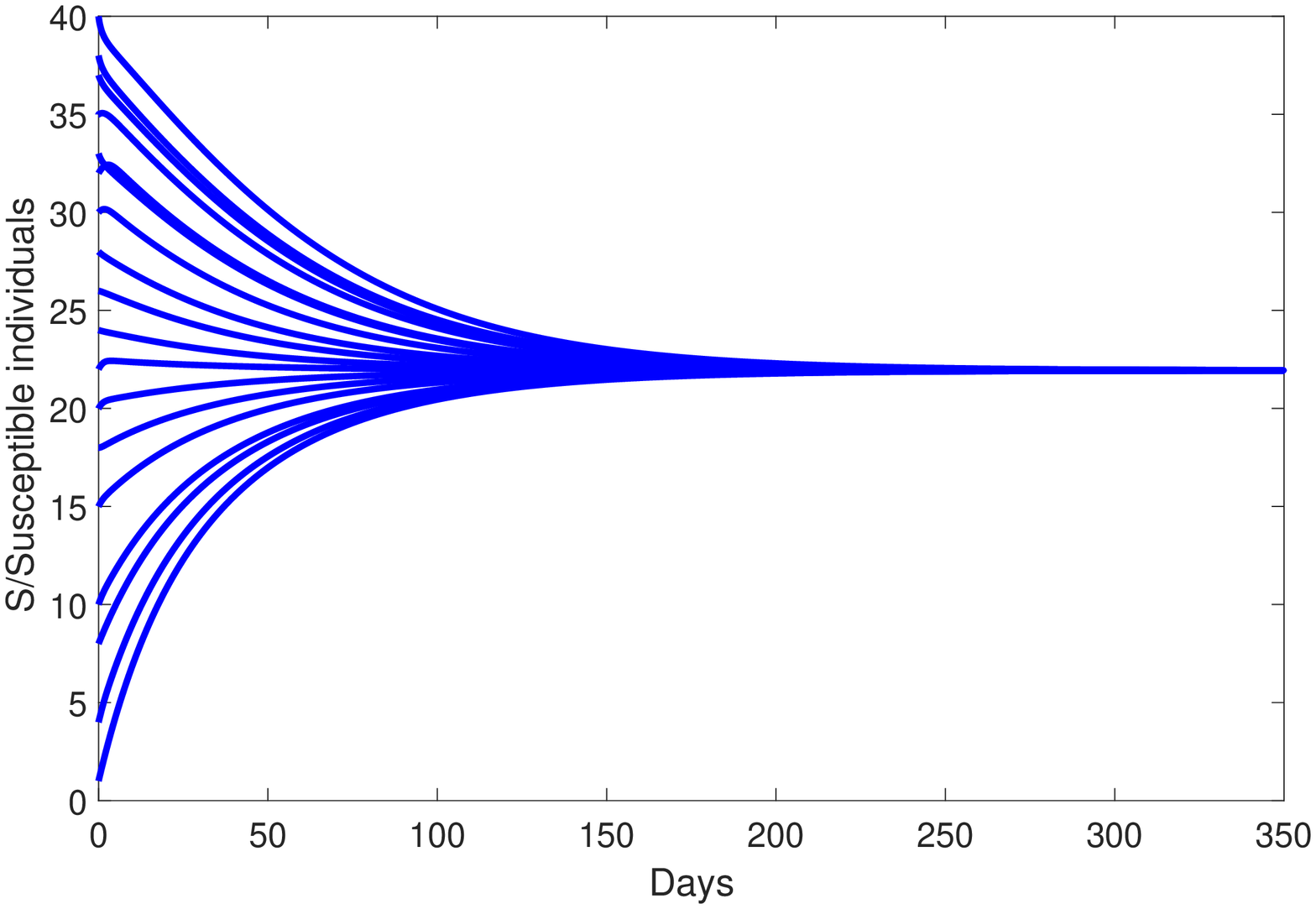}}} {\rotatebox{0}{\includegraphics[width=0.48
							\textwidth, height=40mm]{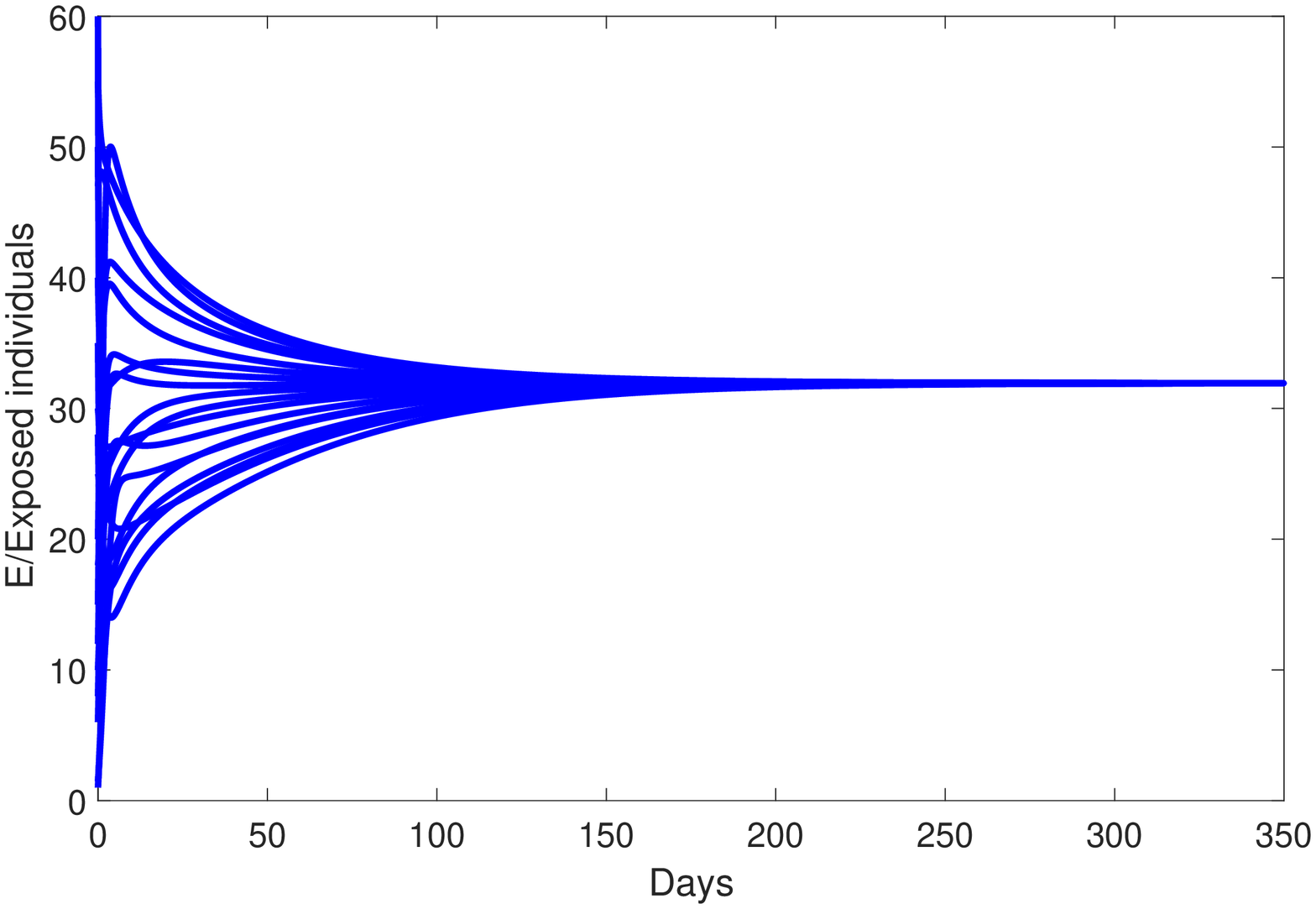}}}
					{\rotatebox{0}{\includegraphics[width=0.48 \textwidth,
							height=40mm]{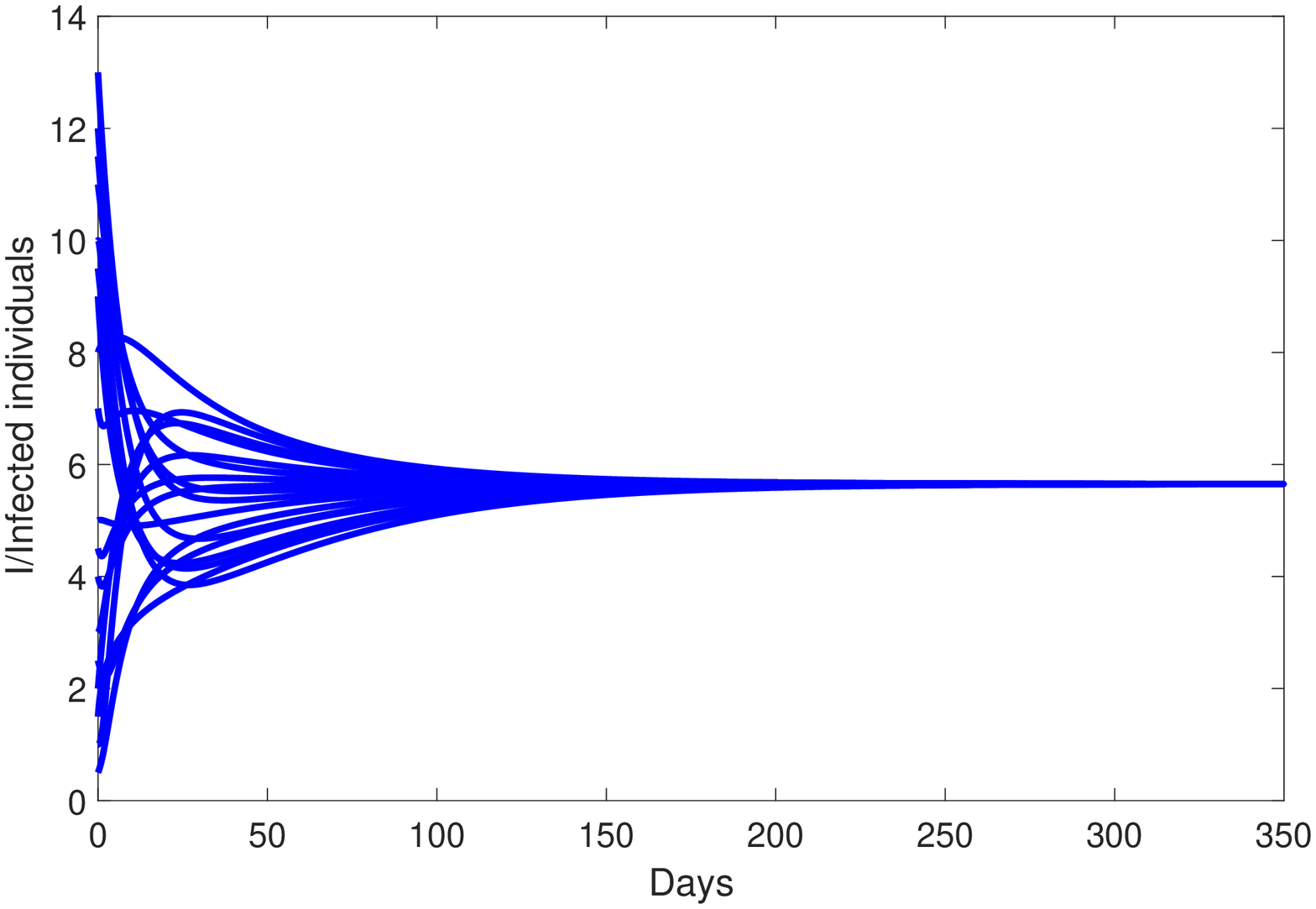}}} {\rotatebox{0}{\includegraphics[width=0.48
							\textwidth, height=40mm]{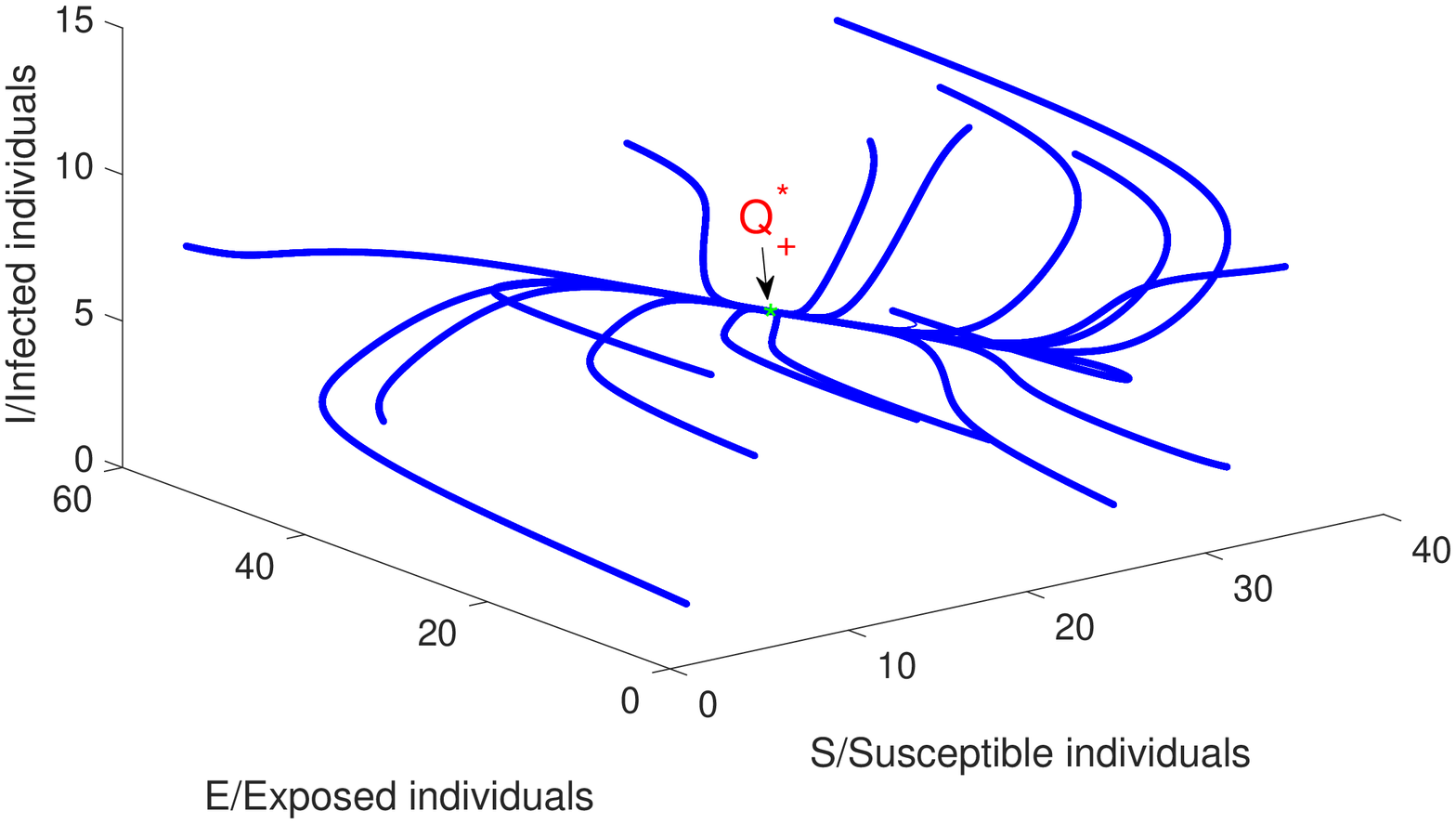}}}
					\caption{
						\footnotesize  Time histories and trajectories of system \eqref{e4} with different initial values for ~$N^*=60, \,k=0.02, \,\mu=0.013, \,~\beta_1=0.0006,\, ~\beta_2=0.0006, ~\alpha_1=0.03, ~\alpha_2=0.04,$ and $ ~\gamma=0.1$. Here, ~$R_0>1$, ~$b_2>0, ~b_1<0,$ and $ ~b_0<0$. We can see that the trajectory of the system converges to $Q_+^*\approx(21.9388, 31.9366, 5.6525)$. Here, $Q_+^*$ is globally asymptotically stable. }\label{F2}
				\end{center}
			\end{figure}
			
			\begin{figure}[!h]
				\begin{center}
					{\rotatebox{0}{\includegraphics[width=0.48 \textwidth,
							height=40mm]{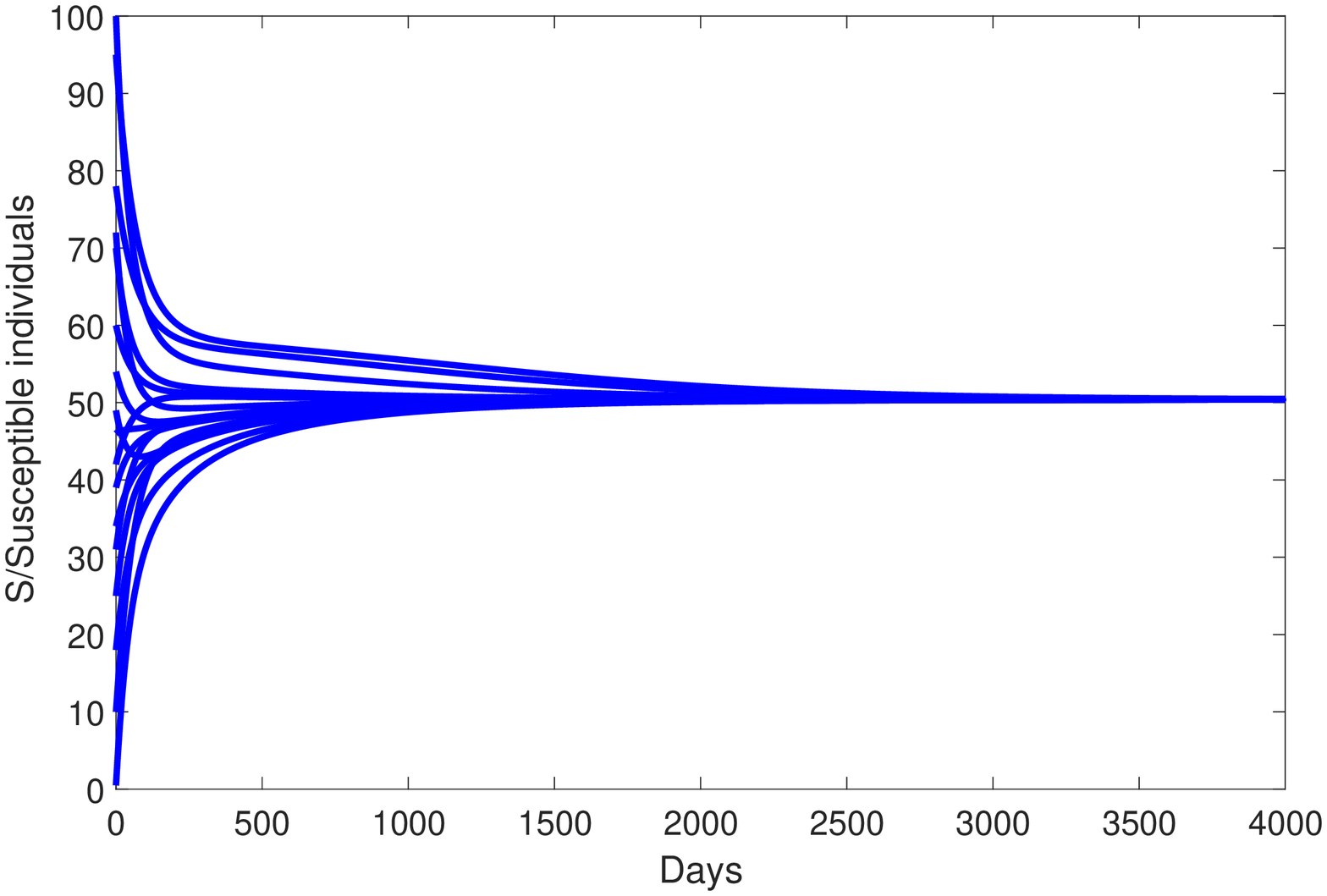}}} {\rotatebox{0}{\includegraphics[width=0.48
							\textwidth, height=40mm]{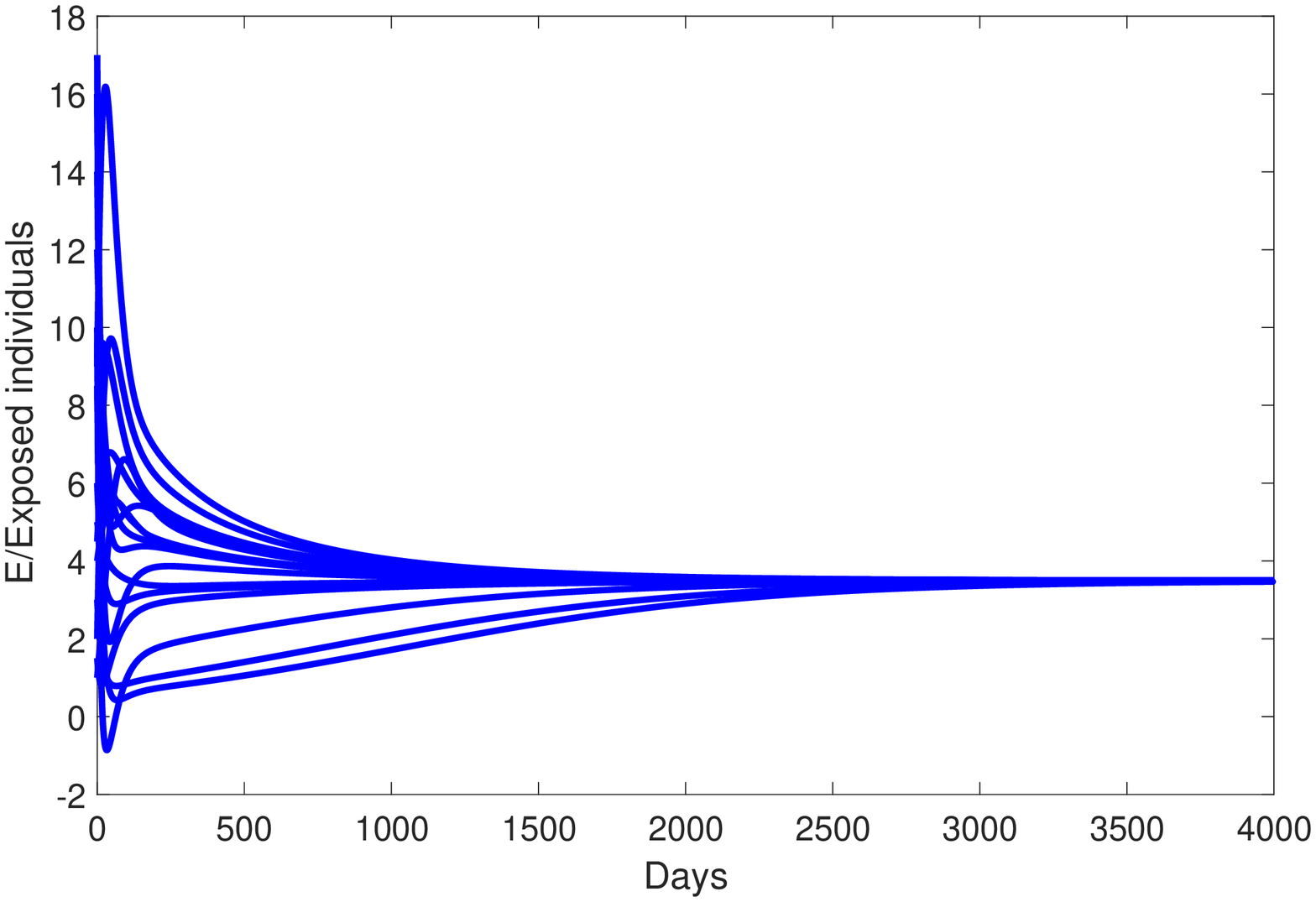}}}
					{\rotatebox{0}{\includegraphics[width=0.48 \textwidth,
							height=40mm]{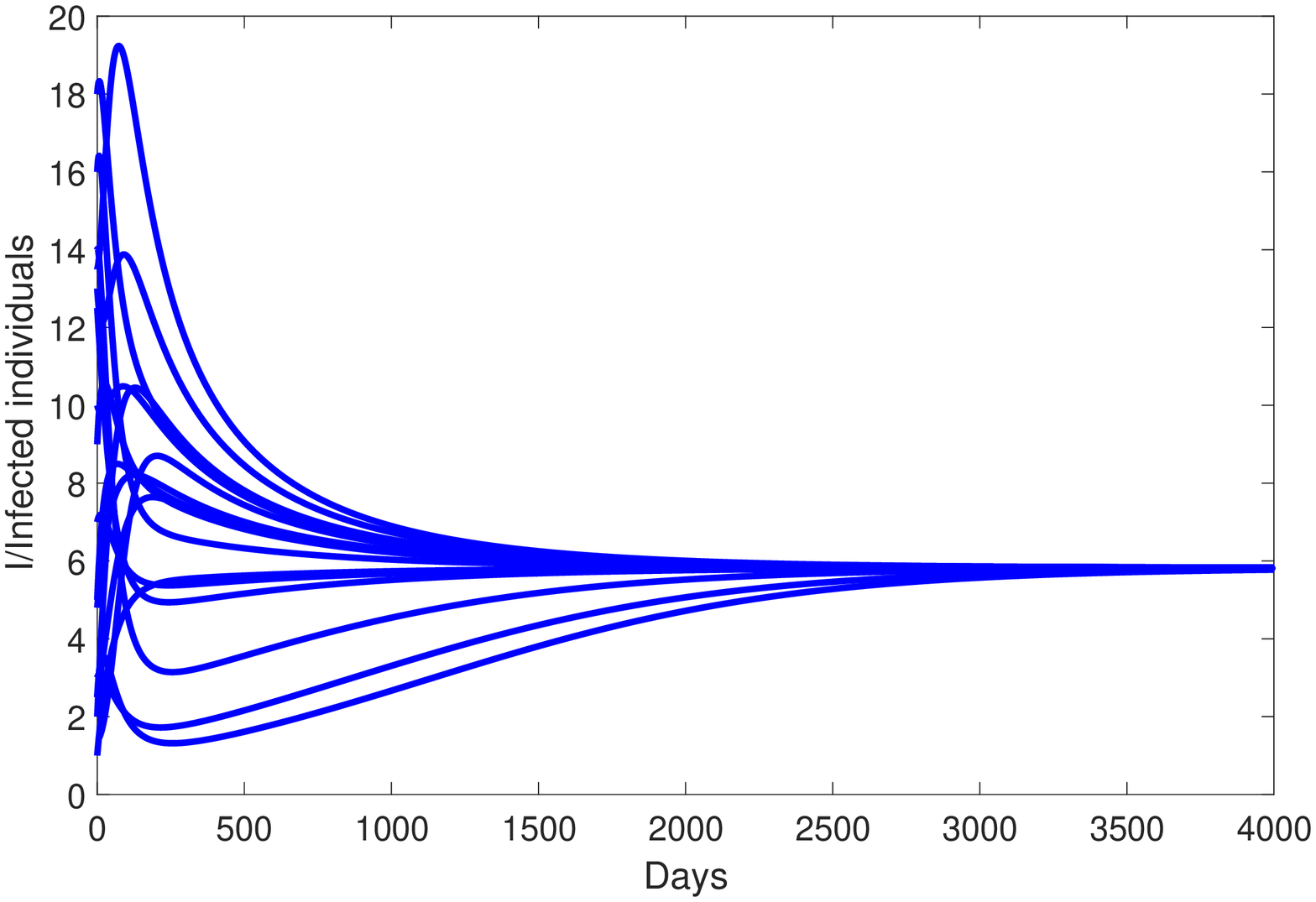}}} {\rotatebox{0}{\includegraphics[width=0.48
							\textwidth, height=40mm]{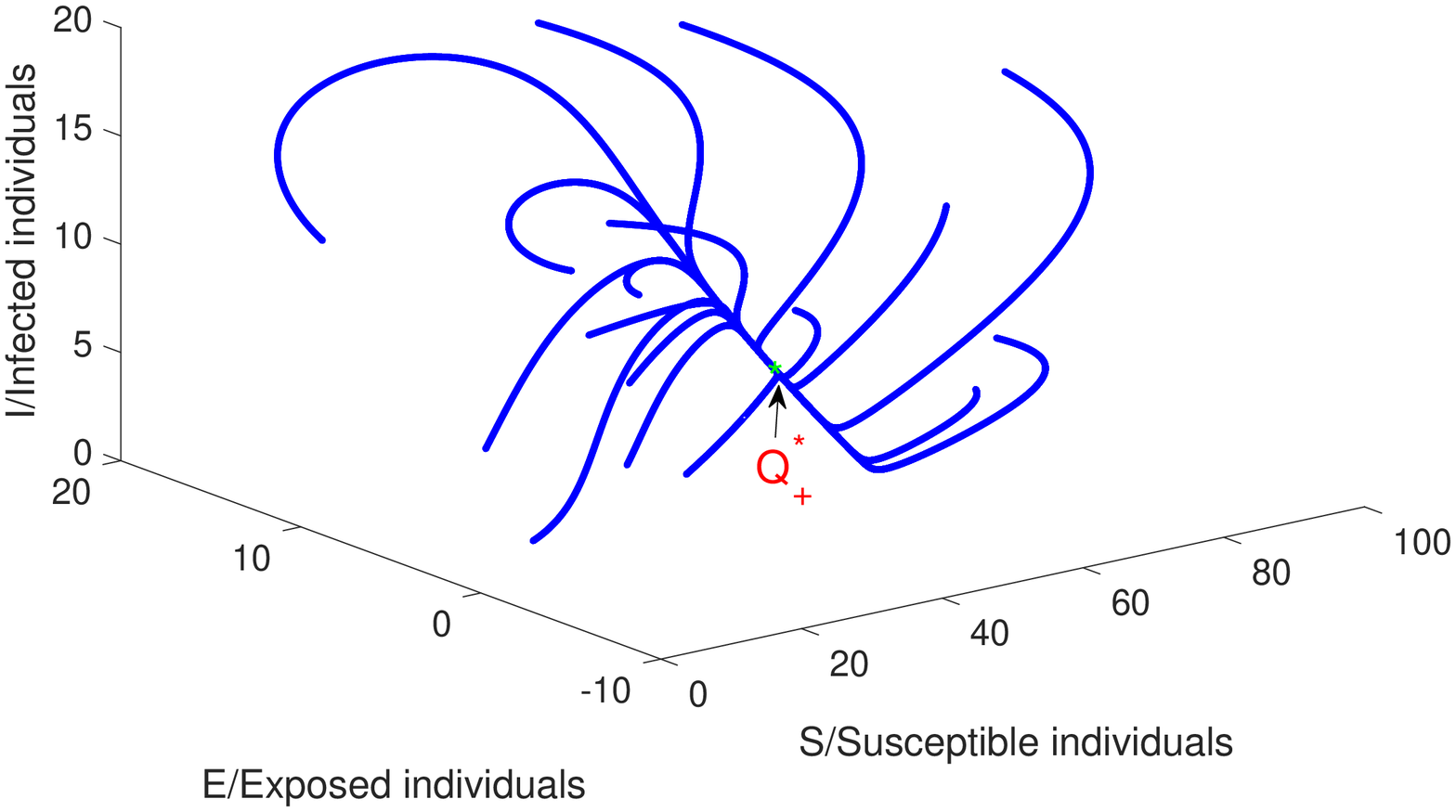}}}
					\caption{
						\footnotesize  Time histories and trajectories of system \eqref{e4} with different initial values for ~$N^*=60, \,k=0.02, \,\mu=0.011, ~\beta_1=0.0001, ~\beta_2=0.0003, ~\alpha_1=0.001, ~\alpha_2=0.001,$ and $ ~\gamma=0.001$. Here, ~$R_0>1$, ~$b_2>0, ~b_1>0,$ and $ ~b_0<0$. We can see that the trajectory of the system converges to $Q_+^*\approx(50.3925, 3.4953, 5.8255)$.  Here, $Q_+^*$ is globally asymptotically stable. }\label{F3}
				\end{center}
			\end{figure}
			
			\begin{figure}[!h]
				\begin{center}
					{\rotatebox{0}{\includegraphics[width=0.48 \textwidth,
							height=40mm]{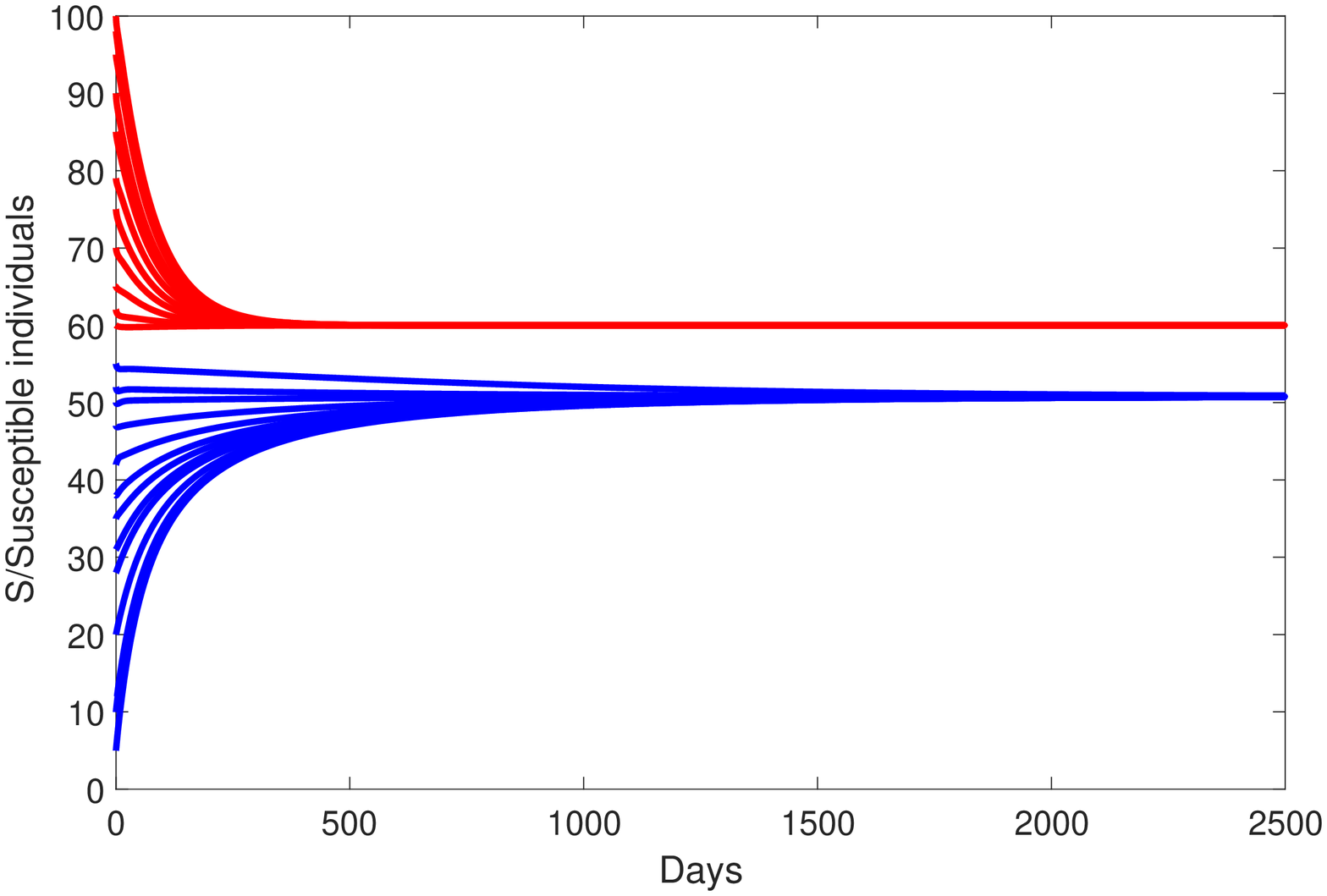}}} {\rotatebox{0}{\includegraphics[width=0.48
							\textwidth, height=40mm]{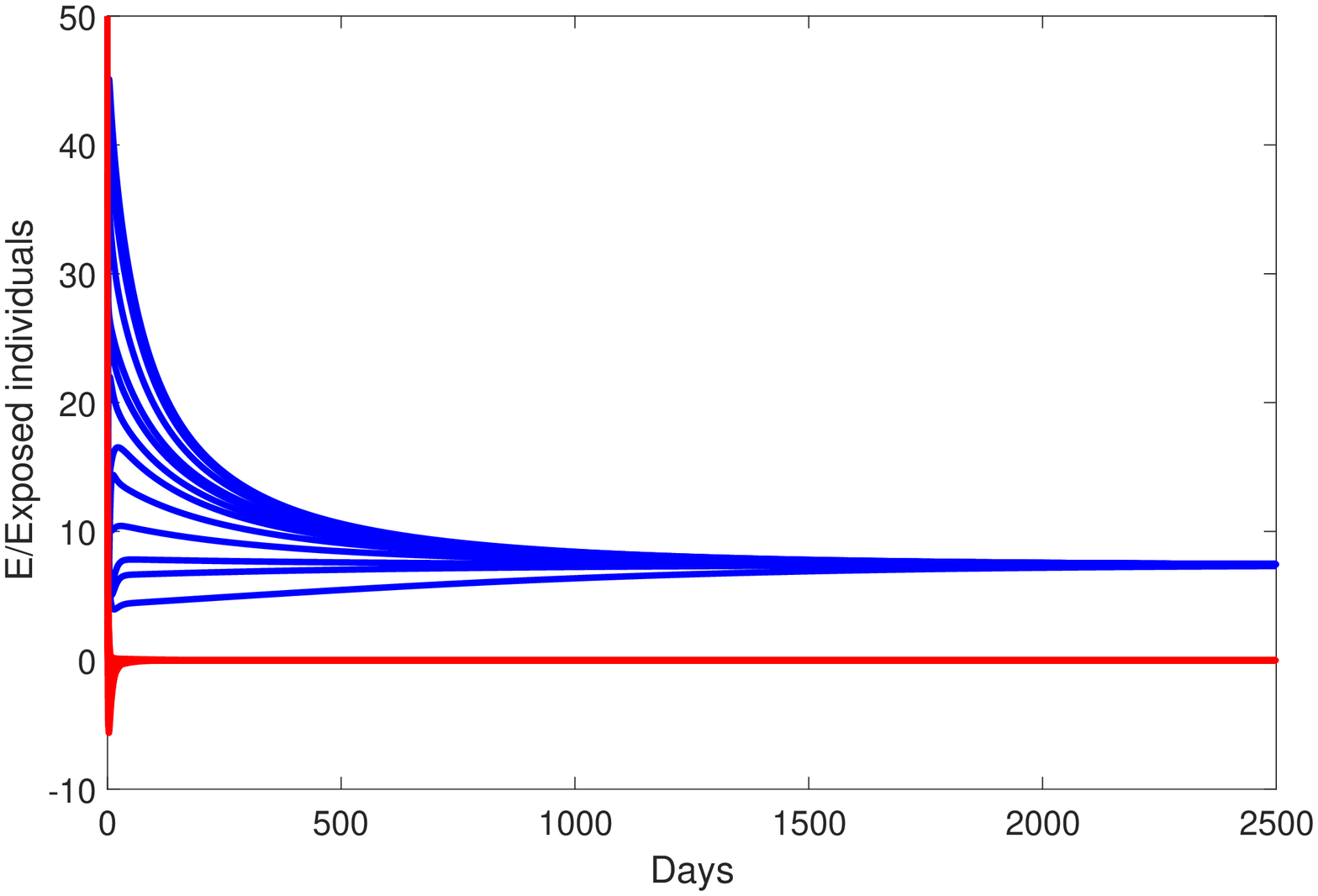}}}
					{\rotatebox{0}{\includegraphics[width=0.48 \textwidth,
							height=40mm]{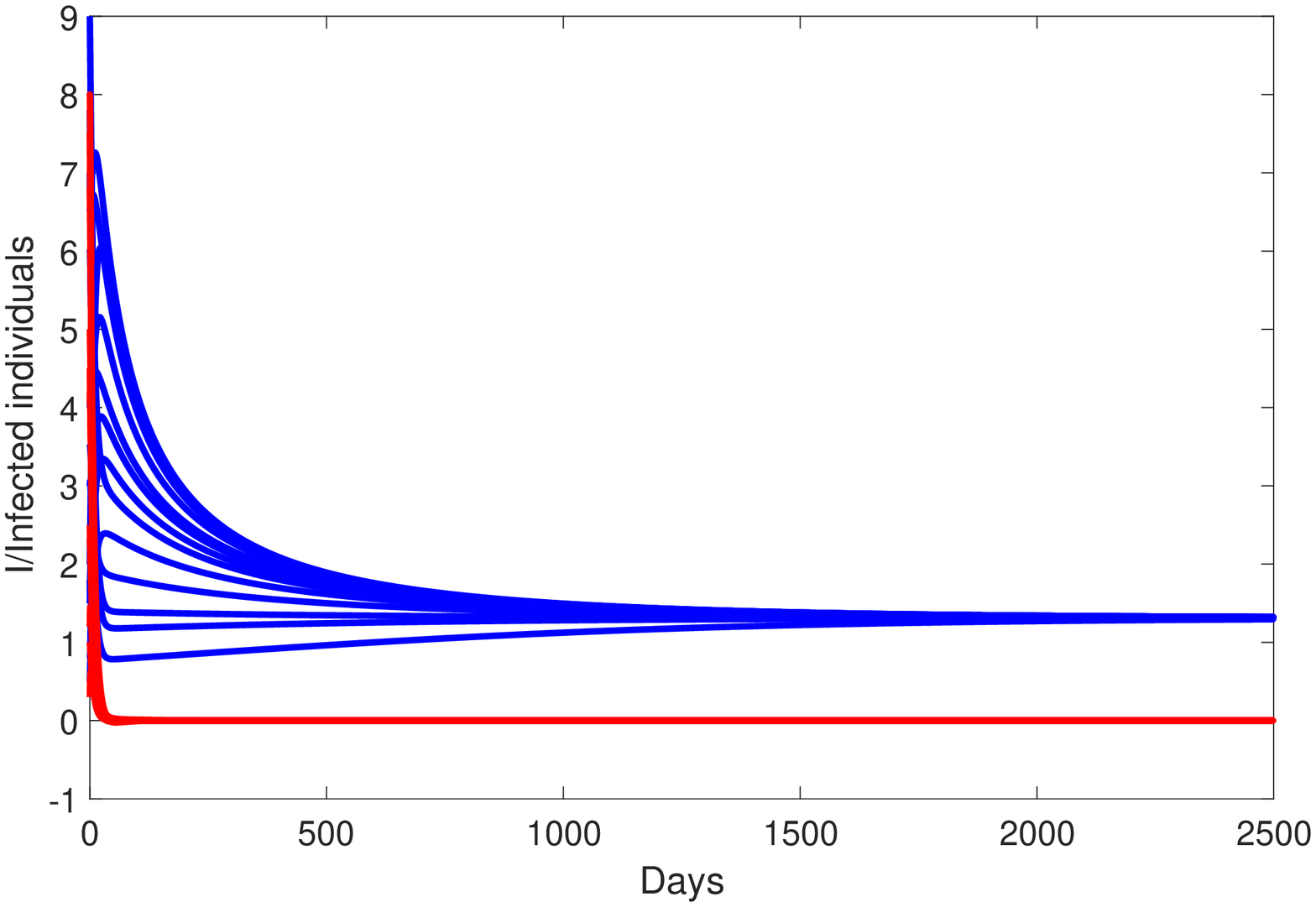}}} {\rotatebox{0}{\includegraphics[width=0.48
							\textwidth, height=40mm]{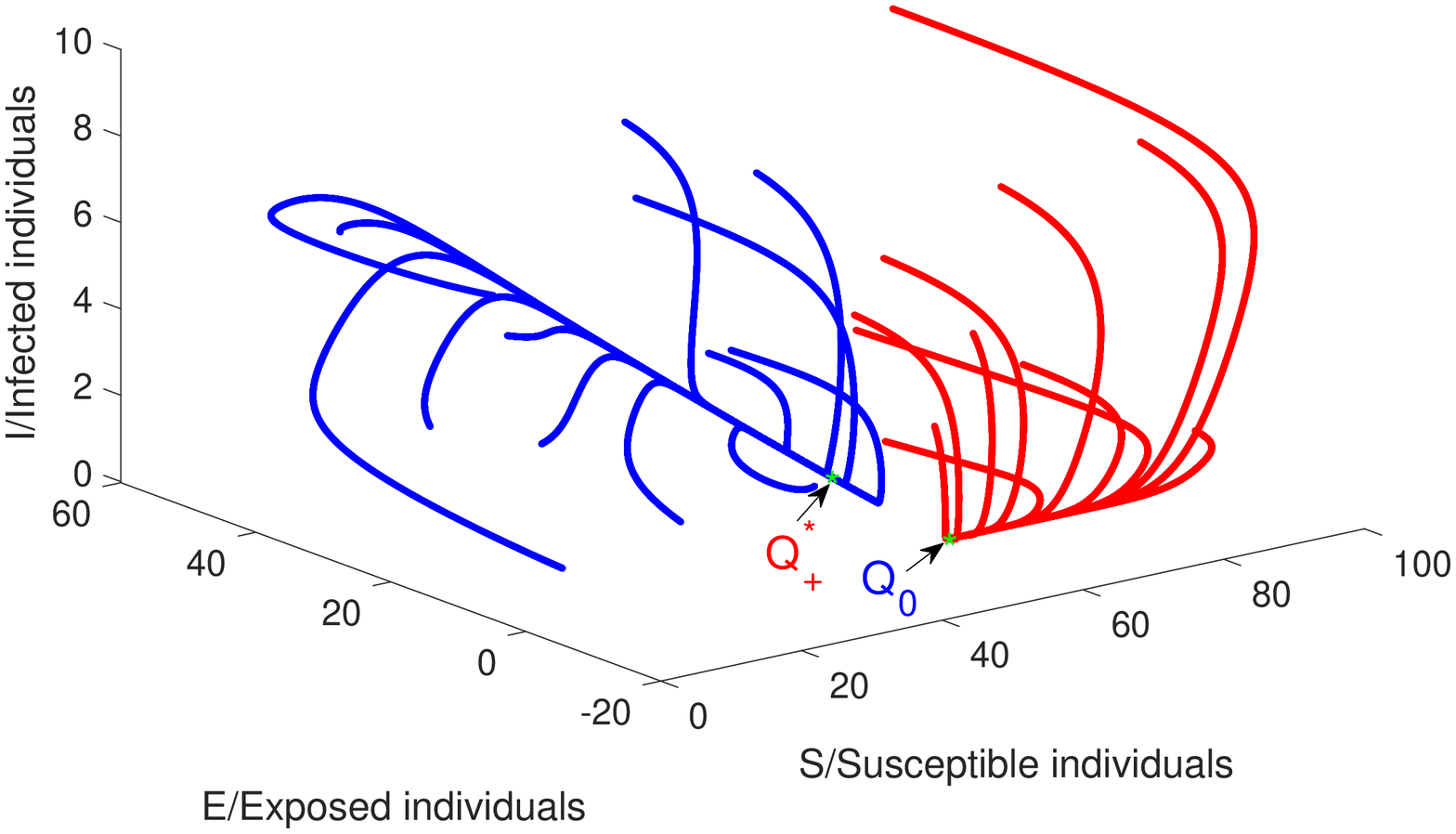}}}
					\caption{
						\footnotesize Time histories and trajectories of system \eqref{e4} with different initial values for ~$N^*=60, \,k=0.02, \,\mu=0.013, \,~\beta_1=0.0003, \, ~\beta_2=0.0001, ~\alpha_1=0.03, ~\alpha_2=0.03, $ and $~\gamma=0.1$. Here,  ~$R_c<R_0<1$, ~$b_2>0, ~b_1<0,$ and $ ~b_0>0$. We can see that the system displays bistability. Both $Q_+^*\approx(50.7976, 7.4129, 1.3120)$ and $Q_0$ are stable.}\label{F4}
				\end{center}
			\end{figure}
			
			\begin{figure}[!h]
				\begin{center}
					{\rotatebox{0}{\includegraphics[width=0.48 \textwidth,
							height=40mm]{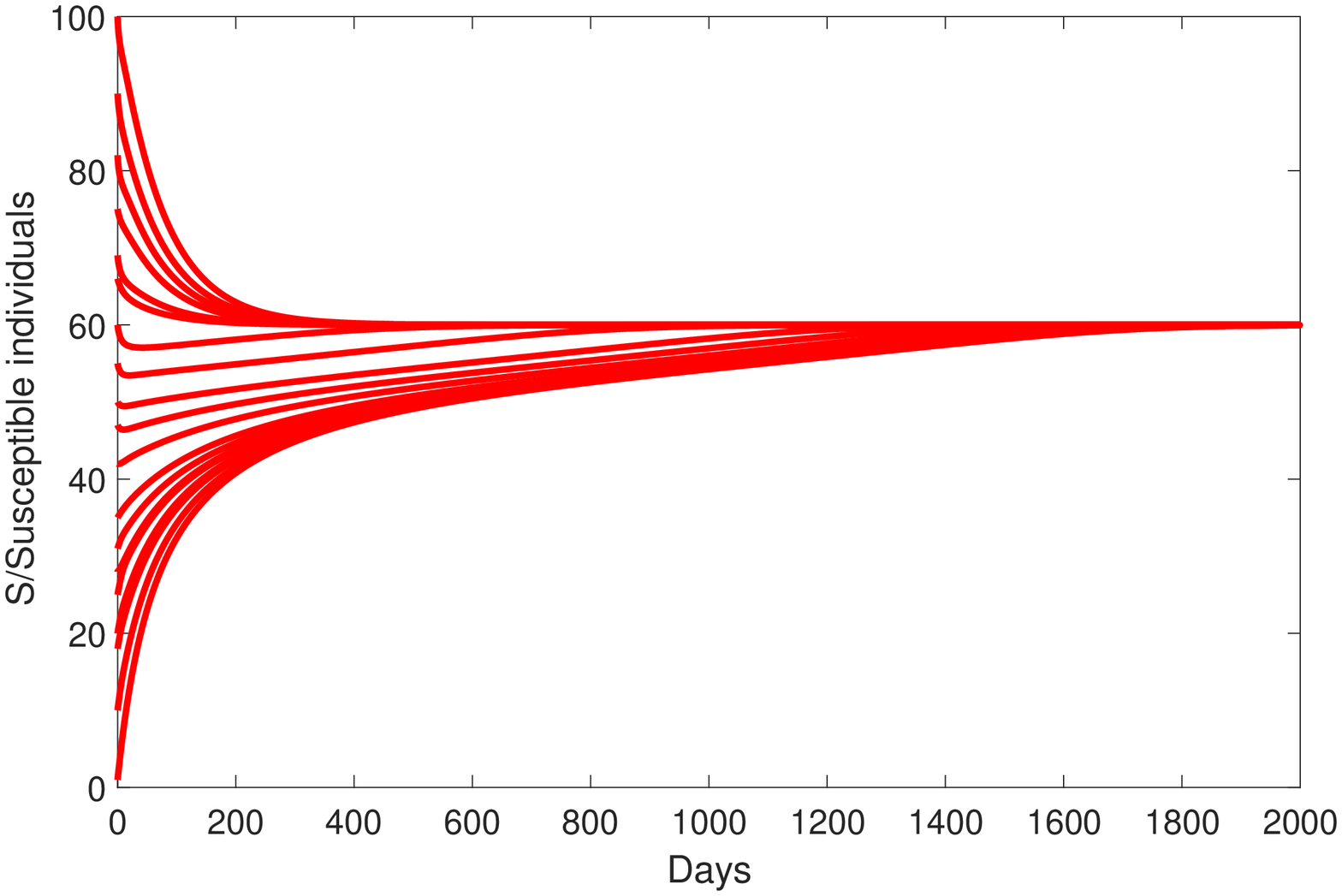}}} {\rotatebox{0}{\includegraphics[width=0.48
							\textwidth, height=40mm]{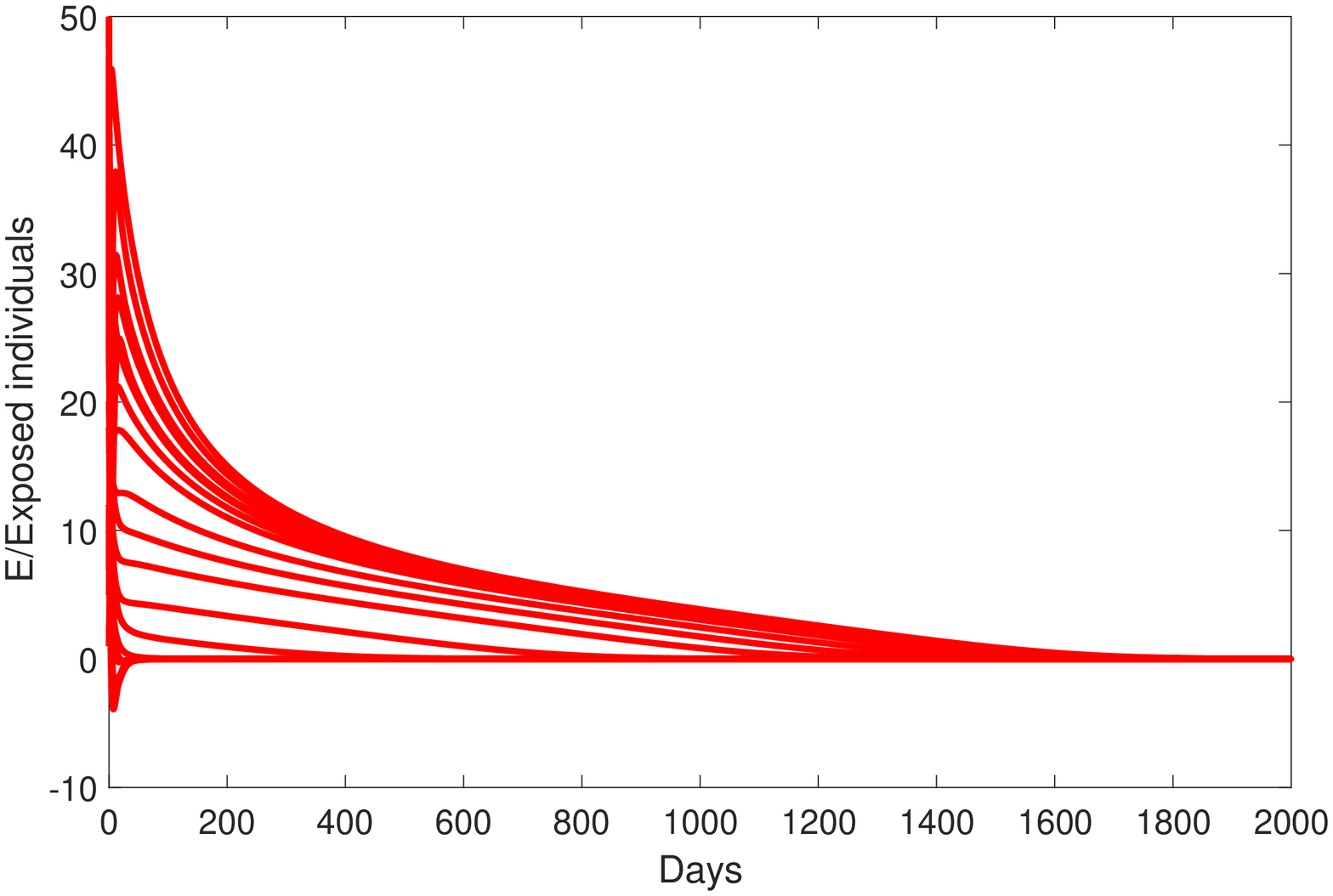}}}
					{\rotatebox{0}{\includegraphics[width=0.48 \textwidth,
							height=40mm]{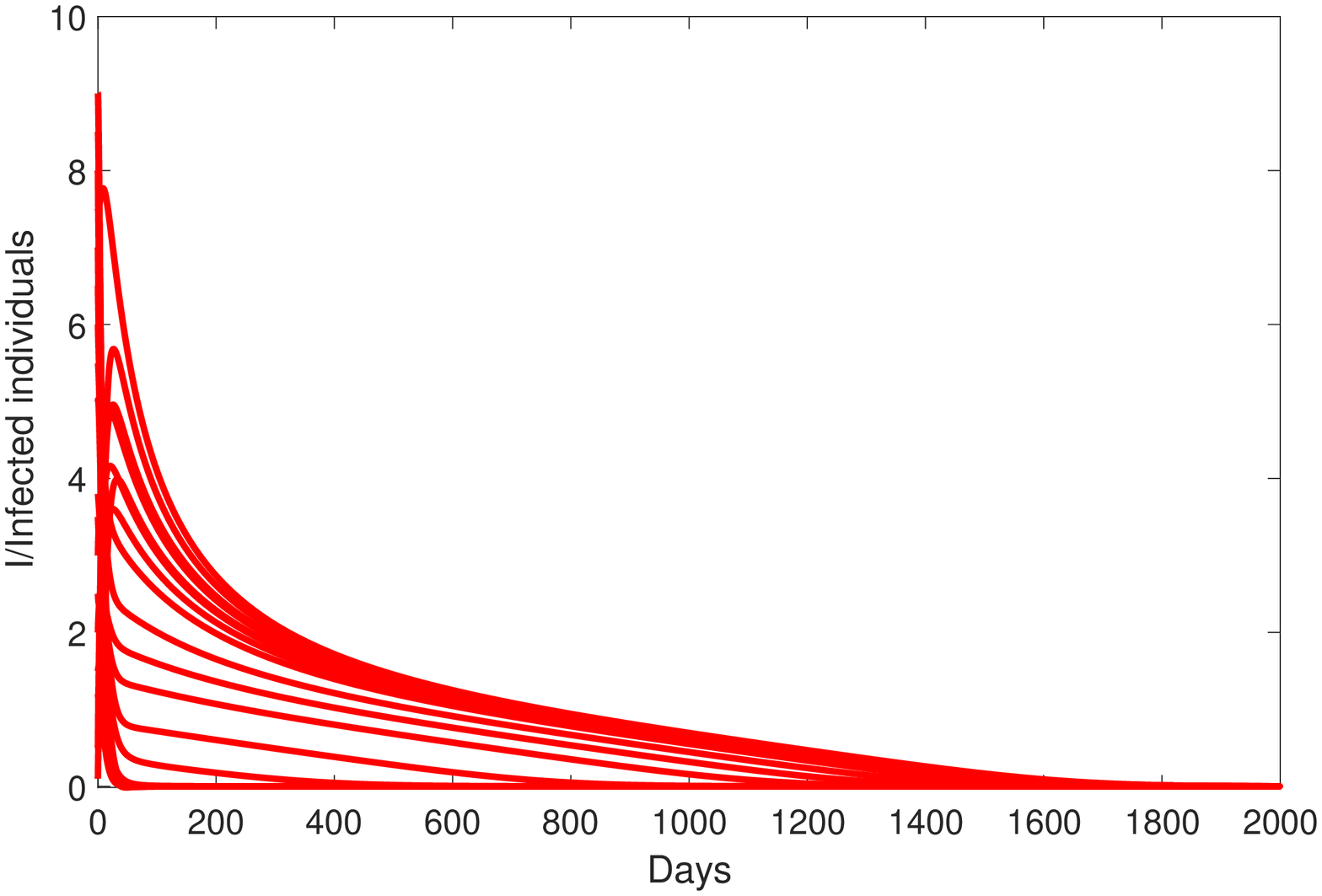}}} {\rotatebox{0}{\includegraphics[width=0.48
							\textwidth, height=40mm]{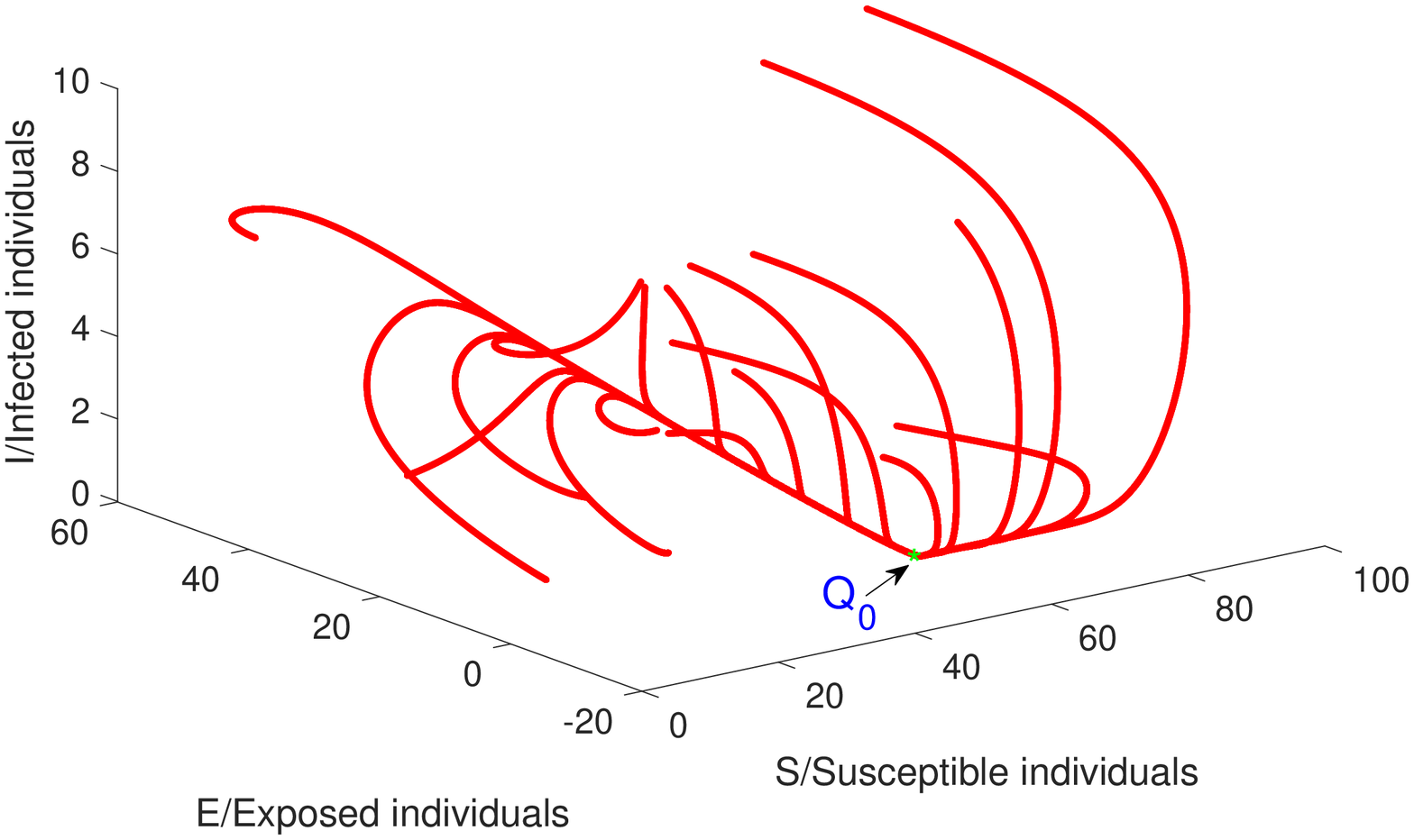}}}
					\caption{
						\footnotesize Time histories and trajectories of system \eqref{e4} with different initial values for ~$N^*=60, \,k=0.02, \,\mu=0.013, \,~\beta_1=0.0003, \, ~\beta_2=0.0001,\, ~\alpha_1=0.01, ~\alpha_2=0.01,$ and $ ~\gamma=0.1$. Here, ~$R_0<R_c<1$, ~$b_2>0, ~b_1<0,$ and $ ~b_0>0$. We can see that the trajectory of the system converges to $Q_0$. Here, $Q_0$ is stable. }\label{F5}
				\end{center}
			\end{figure}
			

			\begin{figure}[!h]
				\begin{center}
					
					\subfigure[]{
						\includegraphics[width=0.48 \textwidth,
						height=40mm]{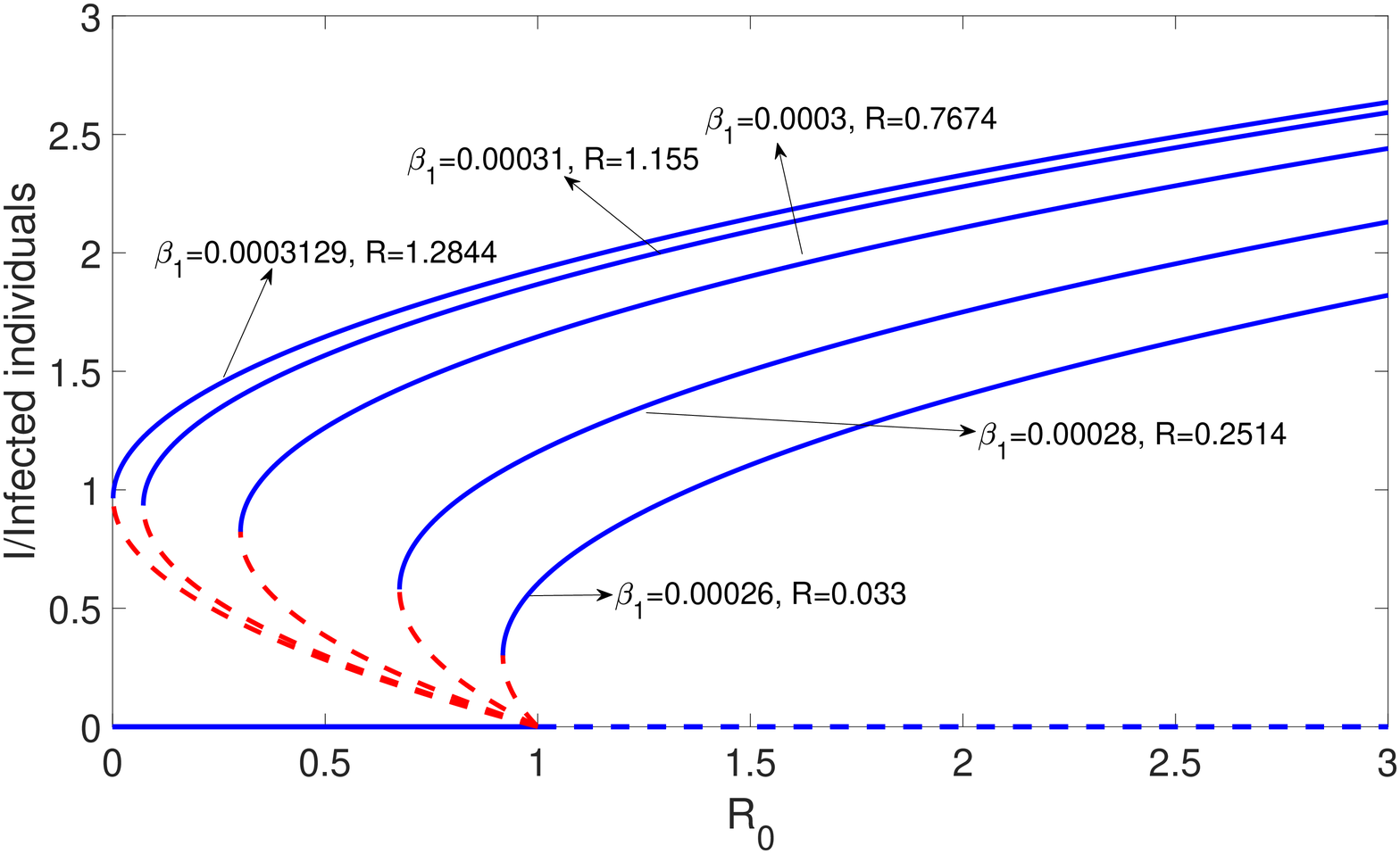}}
					\subfigure[]{
						\includegraphics[width=0.48 \textwidth,
						height=40mm]{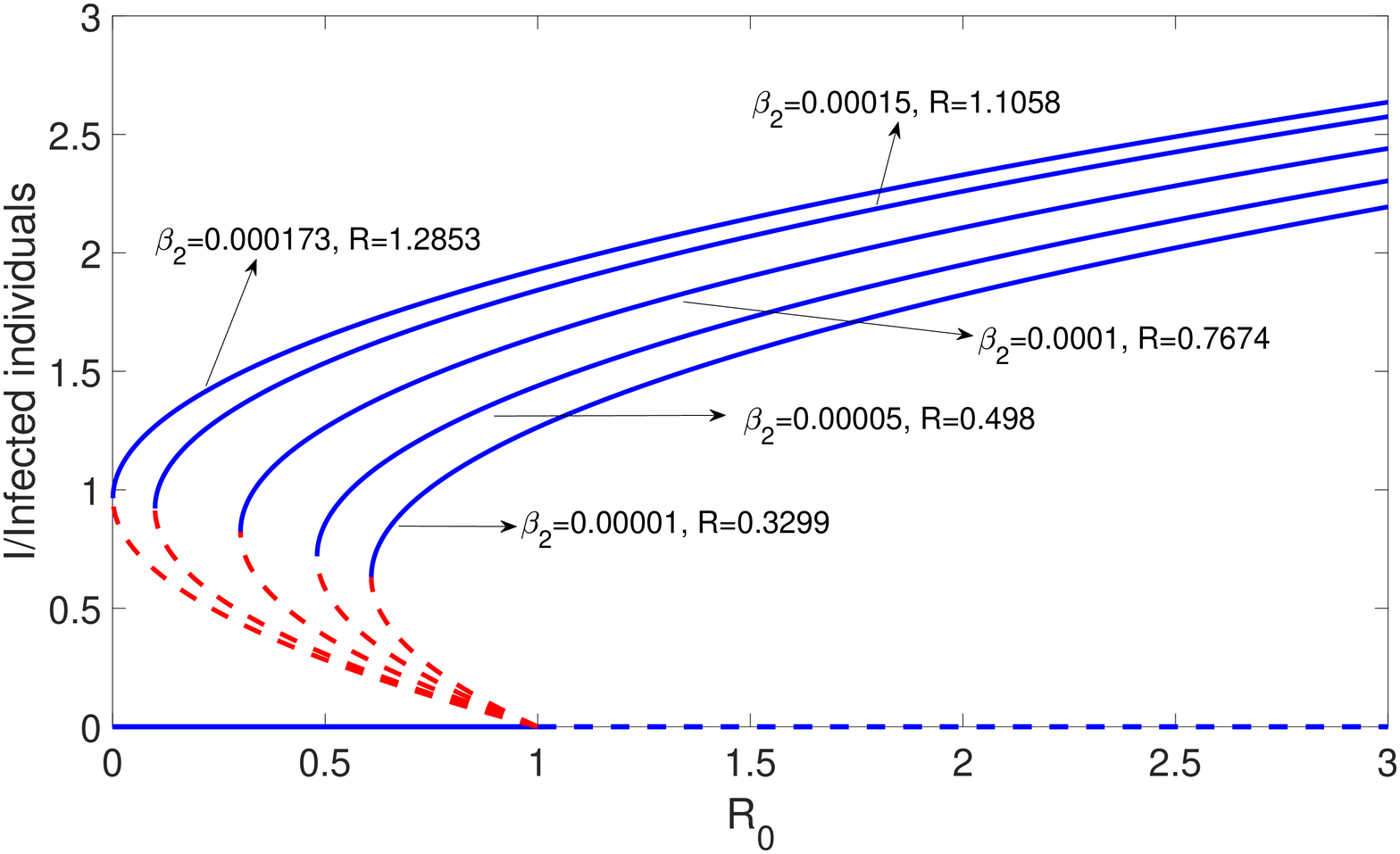}}
					\subfigure[]{
						\includegraphics[width=0.48 \textwidth,
						height=40mm]{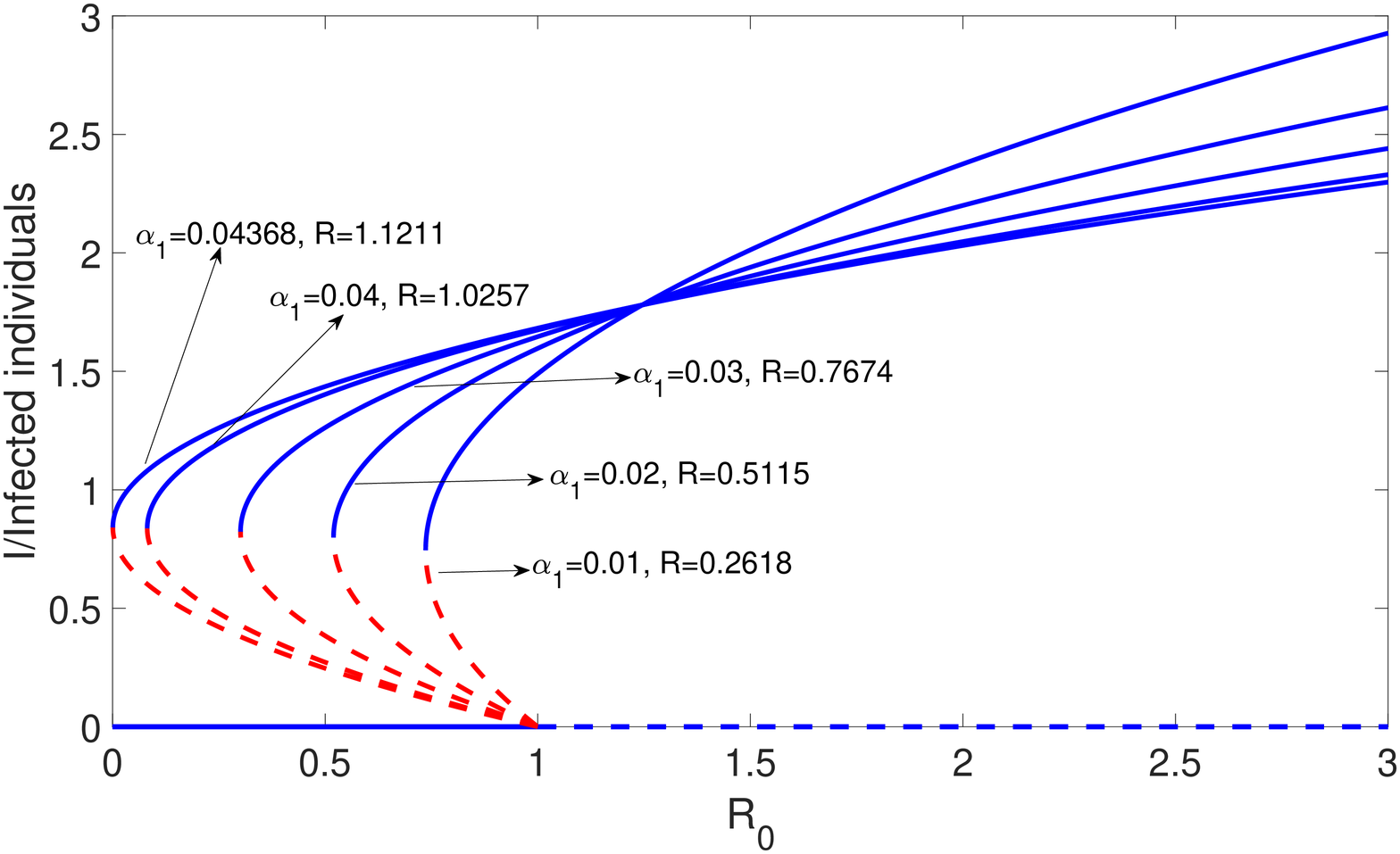}}
					\subfigure[]{
						\includegraphics[width=0.48 \textwidth,
						height=40mm]{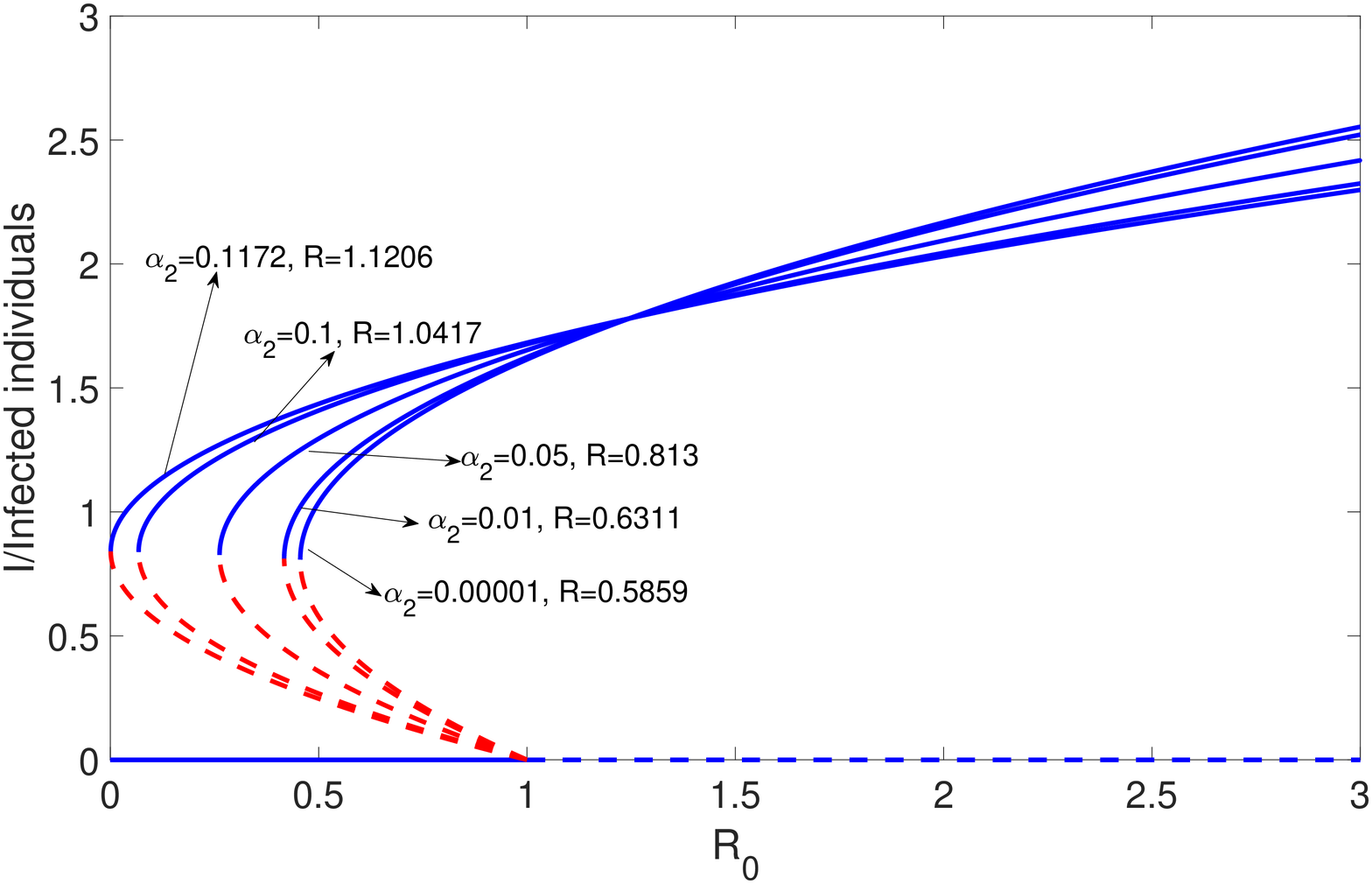}}
					\caption{
						\footnotesize  The effects of changing contact rate $\beta_1$ (a),
						changing contact rate $\beta_2$ (b),  changing contact rate $\alpha_1$ (c) and  changing contact rate $\alpha_2$ (d),
						 on system \eqref{e4}.
						 Here, all other parameter values are the same as those used in Figure \ref{F6}. }\label{F7}
				\end{center}
			\end{figure}

		\end{document}